\documentclass[11pt,draftclsnofoot,onecolumn]{IEEEtran}

\usepackage{amsfonts}
\usepackage{times}
\usepackage{latexsym}
\usepackage{amssymb}
\usepackage{amsmath}
\usepackage{cite}
\usepackage{verbatim}

\newcommand{\bydef}{\triangleq}
\newcommand{\tr}{\it{tr}}

\def\bydef{:=}

\def\bb0{{\mathbb{0}}}

\def\bydef{:=}

\def\bb{{\mathbf{b}}}

\def\bn{{\mathbf{n}}}

\def\br{{\mathbf{r}}}
\def\bs{{\mathbf{s}}}
\def\bt{{\mathbf{t}}}
\def\bu{{\mathbf{u}}}

\def\bx{{\mathbf{x}}}
\def\by{{\mathbf{y}}}
\def\bz{{\mathbf{z}}}
\def\b0{{\mathbf{0}}}

\def\bA{{\mathbf{A}}}
\def\bB{{\mathbf{B}}}

\def\bD{{\mathbf{D}}}

\def\bG{{\mathbf{G}}}
\def\bH{{\mathbf{H}}}
\def\bI{{\mathbf{I}}}

\def\bQ{{\mathbf{Q}}}

\def\bU{{\mathbf{U}}}
\def\bV{{\mathbf{V}}}

\def\bbC{{\mathbb{C}}}

\def\bbE{{\mathbb{E}}}

\def\bydef{:=}

\def\sf0{{\mathsf{0}}}

\usepackage{graphicx}
\usepackage{amssymb}
\usepackage{amsfonts}
\usepackage{amsmath}
\begin{document}
\newtheorem{thm}{Theorem}
\newtheorem{lemma}{Lemma}
\newtheorem{rem}{Remark}
\def\proof{\noindent\hspace{0em}{\itshape Proof: }}
\def\endproof{\hspace*{\fill}~\QED\par\endtrivlist\unskip}
\def\mapright#1{\smash{\mathop{\le}\limits_{#1}}}
\def\mapequal#1{\smash{\mathop{=}\limits_{#1}}}
\title{Capacity Scaling for MIMO Two-Way Relaying}
\author{Rahul~Vaze and Robert W. Heath Jr. \\
The University of Texas at Austin \\
Department of Electrical and Computer Engineering \\
Wireless Networking and Communications Group \\
1 University Station C0803\\
Austin, TX 78712-0240\\
email: vaze@ece.utexas.edu, rheath@ece.utexas.edu}
\date{}
\maketitle
\noindent
\begin{abstract}
A multiple input multiple output (MIMO) two-way relay channel is considered, where two sources want
to exchange messages with each other using multiple relay nodes, and
both the sources and relay nodes are equipped with multiple
antennas. Both the sources are assumed to have equal number of
antennas and have perfect channel state information (CSI) for all
the channels of the MIMO two-way relay channel, whereas, each relay node is
either assumed to have CSI for its transmit and receive channel
(the coherent case) or no CSI for any of the channels (the non-coherent
case). The main results in this paper are on the scaling behavior of
the capacity region of the MIMO two-way relay channel with
increasing number of relay nodes. In the coherent case, the capacity
region of the MIMO two-way relay channel is shown to scale linearly
with the number of antennas at source nodes and logarithmically with
the number of relay nodes. In the non-coherent case, the capacity
region is shown to scale linearly with the number of antennas at the
source nodes and logarithmically with the signal to noise ratio.

\end{abstract}

\section{Introduction}
Relay channels are the most basic building block for cooperative and multihop communication in wireless networks. In a relay channel, one or more nodes, without data of their own to transmit, help a  source destination pair communicate. The origins of the relay channel - as a three terminal communication channel - go back to Van der Meulen \cite{van}. Despite the passage of time, the capacity of even the most basic relay channels is still unknown. Nonetheless, bounds derived in \cite{van, cover1} show that using a relay, it is possible to increase the reliable rate of data transfer between the source and the destination.

Motivated by the capacity improvements obtained by using multiple antennas at
the source and the destination for point-to-point channels \cite{tel},
recently, there has been a significant research focus on
finding the capacity of the multiple input multiple output (MIMO) relay channel,
where the source, the destination, and the relay may have multiple antennas
\cite{wang, caleb, host}.
The capacity of the MIMO relay channel was first studied in \cite{wang,gupta1},
where upper and lower bounds on the capacity of the MIMO relay channel are
derived for the deterministic and the Gaussian fading channel.
Improved lower bounds for the MIMO relay channel with Gaussian fading
channel were provided by \cite{caleb}, where message splitting and
superposition coding are used at the transmitter to improve the bounds
provided in \cite{wang}.
In \cite{wang,caleb} only full-duplex relays (can transmit and receive at the
same time) were considered. Upper and lower bounds on the capacity for the
more practical Gaussian MIMO relay channel with half-duplex relays,
where the relays cannot transmit and receive at the same time,
were developed in \cite{host}. The bounds in
 \cite{wang, caleb, host} indicate that with relays there is a
potential capacity gain to be leveraged by using multiple antennas.

In \cite{van,cover1,wang,caleb, host} only a single source destination pair
is considered with a single relay node.
For a practical wireless network setting, where there are multiple source
destination pairs,
the concept of cooperative communication has been
recently proposed \cite{sen1,sen2,nabar,lane}, where different users in the
network cooperate by taking turns relaying each others data.
Thanks to the spatial separation between users, cooperation between users
provides a means to obtain and exploit spatial diversity gain, called
cooperative diversity gain, which increases the
achievable data rate between each source and its destination.
Several different protocols have been proposed to exploit the cooperative 
diversity gain, e.g. amplify and forward (AF) \cite{sen1,sen2, nabar, lane,gupta1}, decode and forward (DF) \cite{lane1, valenti}, with half-duplex \cite{sriram},
and full-duplex assumptions \cite{host1}.


Prior work on the relay channel mostly considers one-way communication, i.e.
a source wants to send data to a destination.
In most networks, however, the destination also has some data to send to the
source, e.g. packet acknowledgements from the destination to the source, 
downlink and uplink in cellular networks.
Consequently, there has been interest in the two-way relay channel, where the
bidirectional nature of communication is taken into account
\cite{wit,wit2,hol,tar}.
The two-way relay
channel was studied in \cite{wit2}, where upper and lower bounds
on the capacity region were derived for a general
discrete memoryless channel.

The MIMO two-way relay channel
was introduced in \cite{wit}, where two terminals $T_1$
and $T_2$ want to exchange information with each other through a single relay node
as shown in Fig. \ref{twowayarmin} and both $T_1, T_2$ and the relay node is equipped with
multiple antennas. It was assumed in \cite{wit}
that each node can only work in half-duplex mode
and there is no direct path between $T_1$ and $T_2$.
The communication protocol proposed in \cite{wit} for the MIMO two-way relay channel is as
follows. In the first time slot, both $T_1$ and $T_2$ transmit simultaneously and
the relay node receives the superposition of the signals transmitted by $T_1$ and
$T_2$. In the next time slot, the relay node transmits an amplified version
of the signal, received in the last time slot, to both $T_1$ and $T_2$, subject
to a power constraint. Since both $T_1$ and $T_2$
know what they transmitted in the last time slot, both can remove the effect
of their own signal from the received signal, to decode the other terminal's
message. Thus, the MIMO two-way relay channel facilitates simultaneous communication between
$T_1$ and $T_2$ without creating any self interference.
This idea is reminiscent
of network coding \cite{dina}, though note that here the coding is done in analog domain
rather than in digital domain. The MIMO two-way relay
channel is also known by several other names in the literature, namely, bidirectional
MIMO relay channel \cite{tar} and is also a special case of analog network coding
\cite{dina}.

In prior work, achievable rate region (region enclosed by the
rates achievable on the $T_1 \rightarrow T_2$ and $T_2 \rightarrow
T_1$ links, simultaneously) expressions were derived for the
Gaussian half-duplex MIMO two-way relay channel
(fading coefficients as well as additive noise is Gaussian distributed)
using AF \cite{wit} and DF \cite{hol,tar} at the relay node.
A main conclusion derived in prior work \cite{wit,hol,tar}, is that it is possible to
remove the $\frac{1}{2}$ rate loss factor in spectral efficiency due to the
half-duplex assumption on the relay node.
To the best of our knowledge, none of the
achievable rate region expressions for the MIMO two-way relay channel
meet the best known upper bounds \cite{ger} and therefore the capacity region of the MIMO two-way relay channel is unknown.

\begin{figure}
\centering
\includegraphics[height= 1.5in]{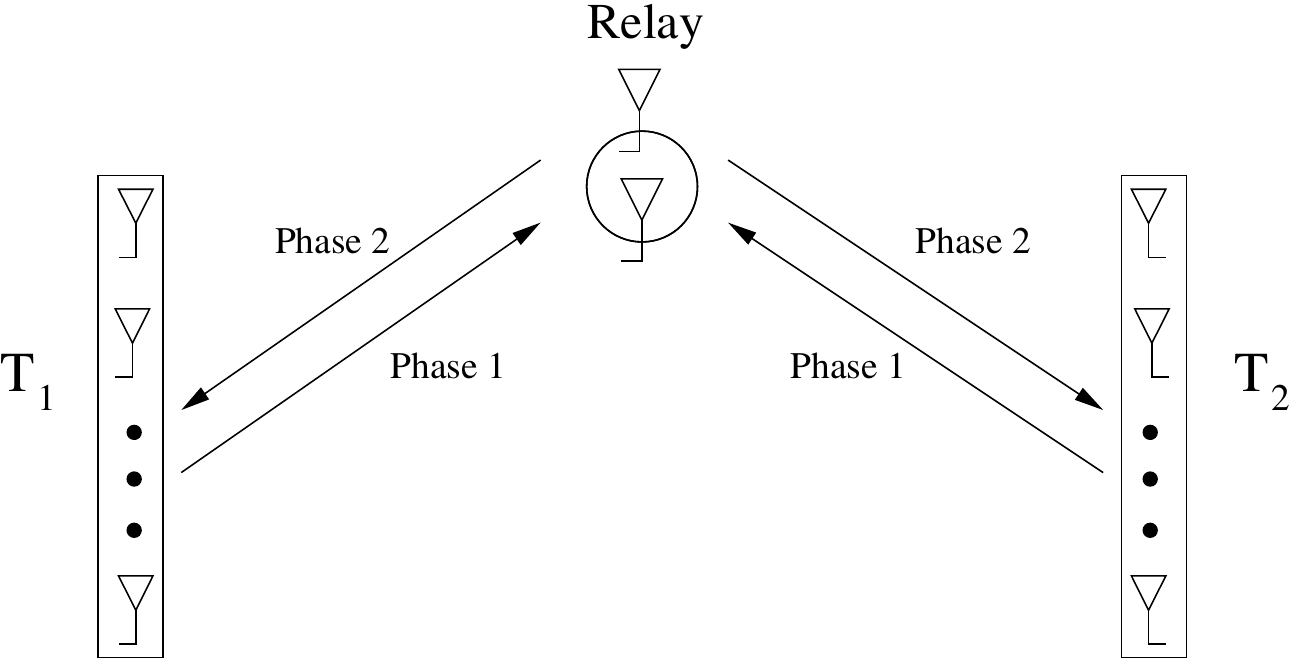}
\caption{Two way relaying Protocol}
\label{twowayarmin}
\end{figure}
%

In this paper we consider a MIMO two-way relay channel with multiple
relay nodes, deriving upper and lower bounds on its capacity
region with different channel state information (CSI) assumptions.
We show that the upper and lower bounds are only
a constant term away, as the number of relays $K$ grows large,
$K \rightarrow$ $\infty$ with probability $1$. Thus, we characterize the
scaling behavior of the capacity region of the MIMO two-way relay channel
as the number of relay nodes grow large.
Our approach is similar to the asymptotic (in the number of relays)
capacity formulation of \cite{gupta, gast, bol}.

Our system model and the key assumptions are as follows.
We assume that two terminals $T_1$ and $T_2$ want to communicate with each other
via $K$ relay nodes. None of the relays have any data of their own and
only facilitate communication between $T_1$ and $T_2$. Both $T_1$
and $T_2$ are equipped with $M$ antennas, while all the $K$ relays
have $N$ antennas each.
We consider a two-phase communication protocol,
where in the any given time slot, for the first  $\alpha, \ \alpha \ \in [0,1]$ fraction of the time slot, both $T_1$ and $T_2$ transmit
simultaneously and all the relays receive. In the rest
$ 1- \alpha$ fraction of the time slot, all the relays simultaneously
transmit and both $T_1$ and $T_2$ receive the signal transmitted by all relays.
We assume that there is no direct path between $T_1$ and $T_2$ and that
$T_1,T_2$ and all the nodes ($T_1$, $T_2$ and all relay nodes) can
only operate in half-duplex mode.
No direct path assumption is reasonable for the case
when relay nodes are used for coverage improvement and the signal strength
on the direct path is very weak. The half-duplex assumption is
made since full-duplex nodes are difficult to realize in practice.
We assume that both $T_1$ and
$T_2$ have perfect CSI for all the
channels of the MIMO two-way relay channel in the receive mode. This could be enabled through
a combination of channel reciprocity and feedback, however, we do
not explore the practicalities of this assumption in this work.
We consider two different assumptions about the availability of CSI at each
relay node. First we consider the case when each relay is assumed to have
perfect CSI for its own transmit and receive channel states,
which is denoted the {\it coherent} MIMO two-way relay channel.
Second we consider the case where the relays are assumed to have no CSI for
any of their channel states, which is denoted the {\it non-coherent} MIMO two-way relay channel.

Under similar assumptions, capacity scaling results have been found in \cite{gast} and \cite{bol} for the one-way relay channel (when $T_2$ has no data for $T_1$).
With a single antenna at both $T_1$, $T_2$ and each relay node, it is
 shown in \cite{gast} that the capacity of the one-way relay channel
scales logarithmically in the number of relay nodes,
as the number of relay nodes grow large.
The capacity scaling result of \cite{gast} was extended in \cite{bol} to the
case where the source and the destination are equipped with $M$ antennas and the all the relay nodes are equipped with $N$ antennas and it was shown that
there is a $M$ fold increase in the capacity compared to the
single antenna nodes \cite{gast}.

The main results in this paper are on the capacity scaling laws
for the MIMO two-way relay channel.
For the coherent MIMO two-way relay channel,
the capacity region is given by the convex hull of
\begin{eqnarray*}
R_{12} & \le & \frac{M}{2}\log{K} + {\cal O}(1) \\
R_{21} & \le & \frac{M}{2}\log{K} + {\cal O}(1)
\end{eqnarray*}
with probability $1$ as $K \rightarrow \infty$, where $R_{12}$ and $R_{21}$ is the rate of information transfer from
$T_1 \rightarrow T_2$ and $T_2 \rightarrow T_1$, and we use the notation $u(x) = {\cal O}(v(x))$  if
$|\frac{u(x)}{v(x)}|$ remains bounded, as $x \rightarrow \infty$.
For this result, the upper bound on the capacity region is obtained for
all $\alpha \in [0,1]$ and an achievable strategy with $\alpha = \frac{1}{2}$
is proposed to achieve the upper bound within a ${\cal O}(1)$ term.

For the non-coherent MIMO two-way relay channel, for a fixed $\alpha=\frac{1}{2}$, i.e. $T_1$ and $T_2$ transmit and receive for same amount of time, the capacity
region is given by the convex hull of
\begin{eqnarray*}
R_{12} & \le & \frac{M}{2}\log{P_R} + {\cal O}(1) \\
R_{21} & \le & \frac{M}{2}\log{P_R} + {\cal O}(1)
\end{eqnarray*}
with probability $1$ as $K \rightarrow \infty$, and $P_R$ is the sum of the
power available at each relay.

The strategy we use for deriving the capacity region of the MIMO
two-way relay channel as $K \rightarrow \infty$ is to obtain an upper bound using the cut-set bound \cite{cover}
and then derive an achievable rate region that approaches the upper bound.
For the coherent case, we propose the following achievable strategy.
Both $T_1$ and $T_2$ transmit $M$ independent data streams from their
$M$ antennas. Each relay node using its CSI, does match filtering for the
channels experienced
by the $M$ data streams from
$T_1\rightarrow T_2$ and $T_2\rightarrow T_1$, simultaneously,
and all the $M$ streams from $T_1$ are decoded jointly at $T_2$ and
vice versa.
We show that this strategy achieves the capacity region upper bound within a
${\cal O}(1)$ term without any cooperation between $T_1$ and $T_2$.
For the non-coherent case, we propose an achievable strategy where
both $T_1$ and $T_2$ transmit $M$ independent data streams from their
$M$ antennas. Since none of the relays have any CSI in this case,
we propose an AF achievable strategy where each relay transmits a
scaled version of received signal subject to its power constraint, similar to
\cite{bol}. With this strategy, as $K \rightarrow \infty$,
the channel between $T_1 \rightarrow T_2$
and $T_2 \rightarrow T_1$ converges to an $M \times M$ matrix with
independent and identically distributed entries that are Gaussian distributed
 and we show that the achievable rate region
provided by this AF strategy is within a ${\cal O}(1)$ term of the
upper bound in the high signal-to-noise
ratio (SNR) regime.

From an analytical perspective our work is closely related to \cite{bol},
which only deals with MIMO one-way relay channel. We summarize
the key differences and improvements of the proposed work compared to the
MIMO one-way relay channel capacity scaling result of \cite{bol} as follows.
\begin{itemize}
\item We assume a sum power constraint across
all the relays, which is a generalization of the individual power
constraint considered in \cite{bol}. An individual power constraint
might seem more reasonable from a practical point of view. We show
that even with a sum power constraint, however, the upper bound on
the capacity region of both the coherent and non-coherent MIMO
two-way relay channel can be achieved by allocating equal power to
all relay nodes. Thus, with a sum power constraint, as $K
\rightarrow \infty$, the total power transmitted by all relay nodes
remains bounded as opposed to \cite{bol}, where it is unbounded. The
optimal power allocation is similar to \cite{bol} from a practical
perspective, since each relay node is required to transmit the same
amount of power.

\item We upper bound the capacity of the coherent MIMO two-way
relay channel over all possible two-phase protocols, i.e. over
arbitrary $\alpha$, while in \cite{bol} an upper bound is derived only for
 $\alpha=\frac{1}{2}$.

\item Our achievable AF strategy for the coherent MIMO two-way relay channel
allows all the relays to help all the data streams going from $T_1$
and $T_2$ and $T_2$ to $T_1$ as opposed to \cite{bol} where only
$K/M$ relays are allowed to help each data stream. Moreover, in our
AF strategy joint decoding is performed at both the receivers in
contrast to \cite{bol}, where each data stream is decoded by a
single receive antenna treating all other streams as interference.
Due to both these advantages, our AF strategy provides with better
achievable rate regions compared to \cite{bol} for any finite $K$
and a better ${\cal O}(1)$ term as $K \rightarrow \infty$.

\item For the non-coherent MIMO two-way relay channel, we derive
an upper and lower bound for $\alpha = \frac{1}{2}$,
which differs by only a constant term at high SNR,
while in \cite{bol} only an achievable AF strategy is provided
without any upper bound for $\alpha = \frac{1}{2}$.
\end{itemize}

Our results show that with the MIMO two-way relay channel
there is a improvement in the capacity scaling by a factor of $2$, compared to
MIMO one-way relay channel \cite{bol},
for both the coherent and the non-coherent case.
We show that with the MIMO two-way relay channel, both $T_1$ and $T_2$
can simultaneously communicate with each other at a
rate which is equal to the maximum rate at which $T_1$ can
communicate to $T_2$ if $T_2$ was silent.
Therefore as $K \rightarrow \infty$,
the MIMO two-way relay channel is shown to create two interference free
parallel channels, one for $T_1 \rightarrow T_2$
and another for $T_2 \rightarrow T_1$, where on each channel a
rate given by the maximum possible
rate at which $T_1$ can communicate to $T_2$ link if $T_2$ was silent (one-way
communication \cite{bol})
is achievable.

{\it Organization:}
The rest of the paper is organized as follows. In Section \ref{sys},
we describe the MIMO two-way relay channel system model, the protocol under consideration and the key assumptions. In Section \ref{upbound}, we derive an upper bound on the capacity of the coherent MIMO two-way relay channel.
In Section \ref{ach}, by using a simple combining operation at the relays, we derive the asymptotic achievable  rate region for the coherent MIMO two-way
relay channel and show that it is possible to achieve the upper bound on
the capacity region of the coherent MIMO two-way relay channel within a
${\cal O}(1)$ term.
Section \ref{disc} summarizes and discusses the implication of the coherent
MIMO two-way relay channel capacity region.
For the non-coherent MIMO two-way relay channel, in Section
\ref{noncohup} we derive an upper bound on the achievable rate region.
Section \ref{noncohaf} gives a result on
asymptotic achievable rate region for the non-coherent
MIMO two-way relay channel using AF strategy at relays.
We draw some final conclusions in Section \ref{conc}.

{\it Notation:}
The following notation is used in this paper.
The superscripts $^T, ^*$ represent the transpose and transpose conjugate.
${\bf M}$ denotes a matrix, ${\bf m}$ a vector and $m_i$ the $i^{th}$ element
of ${\bf m}$. For a matrix ${\bf M} = [{\bf m}_1 \ {\bf m}_2 \ \ldots \ {\bf m}_n]$ by
${\text vec}({\bf M})$ we mean $[{\bf m}^T_1 \ {\bf m}^T_2 \  \ldots \ {\bf m}^T_n]^T$.
$det({\bf M})$ and $tr({\bf M})$ denotes the determinant and
trace of matrix ${\bf A}$, respectively.
${\bbE}_{x}(f(x))$ denotes the expectation of function
$f$ with respect to $x$.
$|| \cdot ||$ denotes the usual Euclidean norm of a vector. ${\bf I}_m$ is
a $m\times m$ identity matrix. $|{\cal X}|$ is the cardinality of set ${\cal X}$.
We use the usual notation for $u(x) = {\cal O}(v(x))$  if
$\left|\frac{u(x)}{v(x)}\right|$ remains bounded, as $x \rightarrow \infty$.
A circularly symmetric
complex Gaussian random variable with zero mean and variance $\sigma$
is denoted by $x \sim {\cal CN}(0,\sigma)$ and
$x|y \sim {\cal CN}(0,\sigma)$ denotes that given $y$,
$x$ is a circularly symmetric
complex Gaussian random variable with zero mean and variance $\sigma$.
The variance of a random variable $a$ is denoted by $\text{var}(a)$.
${\bbC}^{MN}$ denotes the set of $M\times N$ matrices with complex entries.
$x_n \xrightarrow{w.p. 1} y$ denotes that the sequence of random variables
$x_n$ converge to a random variable $y$ with probability $1$. We use $a \mapequal{w.p.1} b$ to denote equality with probability $1$ i.e. $Prob.(a=b) =1$  and
 $\mapright{w.p.1}$ is defined similarly.
$I(x;y)$ denotes the mutual information between $x$ and $y$ and $h(x)$ the differential entropy of $x$ \cite{cover}. To define a variable we use the symbol
$\bydef$.

\section{System and Channel Model}
\label{sys}
In this section we describe the MIMO two-way relay channel communication
protocol, followed by signal and channel models.
\begin{figure}
\centering
\includegraphics[height= 3in]{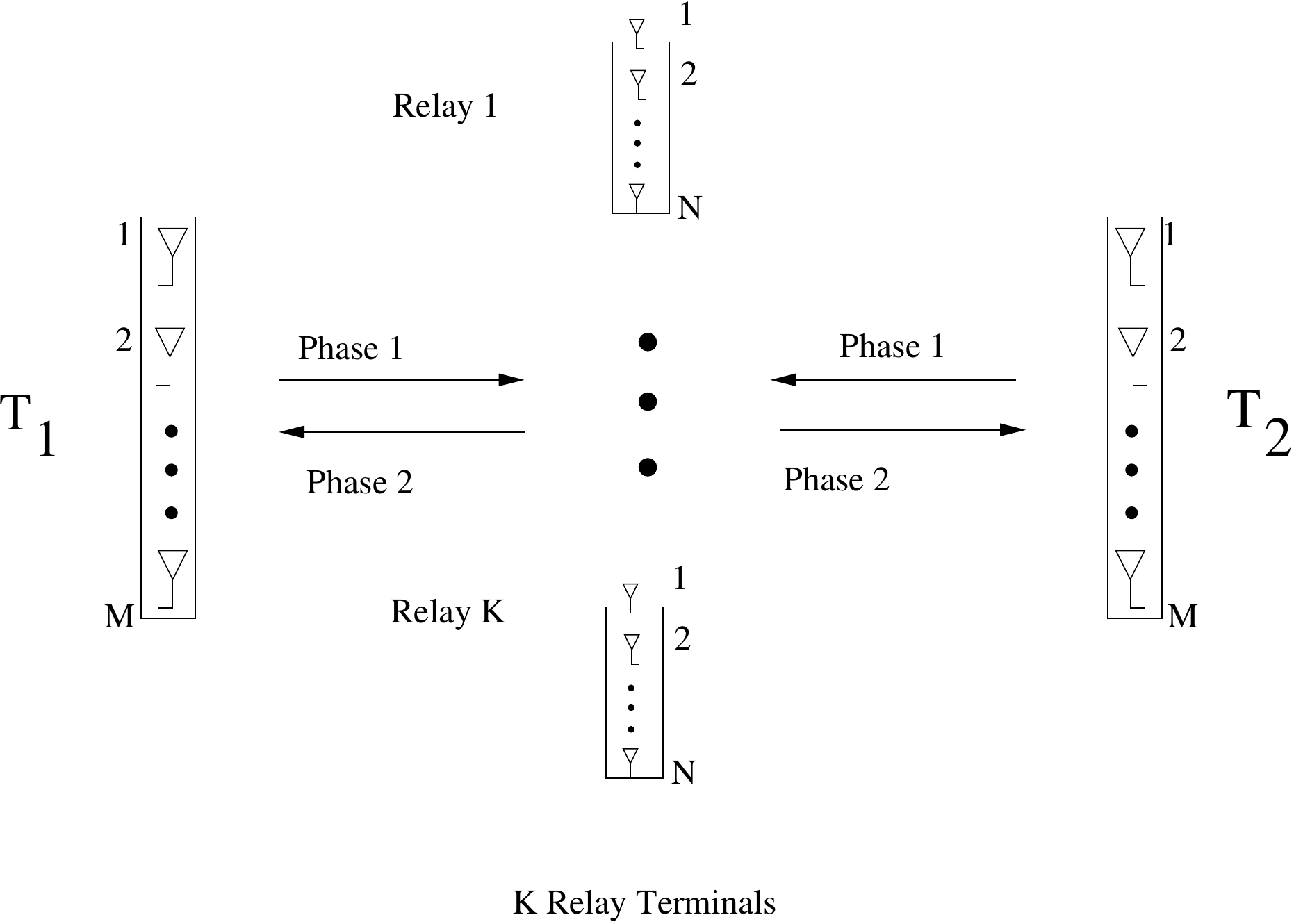}
\caption{Two way communication in two hops}
\label{twoway}
\end{figure}
Consider a wireless network where there are two terminals $T_1$ and
$T_2$ who want to exchange information via $K$ relays, as shown
in Fig. \ref{twoway}. The $K$ relays do not have any data of their
own and only help $T_1$ and $T_2$ communicate.
We assume that there is no direct path between $T_1$ and $T_2$ and
that they can communicate only through the $K$ relays. This is a
realistic assumption when relaying is used for coverage improvement
in cellular systems, since at the cell edge the signal to noise ratio is extremely low for the direct path. In ad-hoc networks, this occurs when two terminals want to communicate, but are out of each other's
transmission range.

We assume that both the terminals $T_1$ and $T_2$ have $M$ antennas
while all the $K$ relays each have $N$ antennas.
The terminals $T_1,T_2$ and all the relays operate in half-duplex mode i.e. cannot transmit and receive at the same time.
The communication protocol is summarized as follows  \cite{wit}.
In any given time slot, for the first $\alpha$ fraction of time, called the
{\it transmit phase}, both $T_1$ and $T_2$ are scheduled to transmit and
all the relays receive a superposition of the signals transmitted from $T_1$
and $T_2$. In the rest $(1-\alpha)$ fraction of the time slot, called the {\it receive phase},
all the relays are scheduled to
transmit simultaneously and both the terminals receive.

\begin{figure}
\centering
\includegraphics[height= 1.5in]{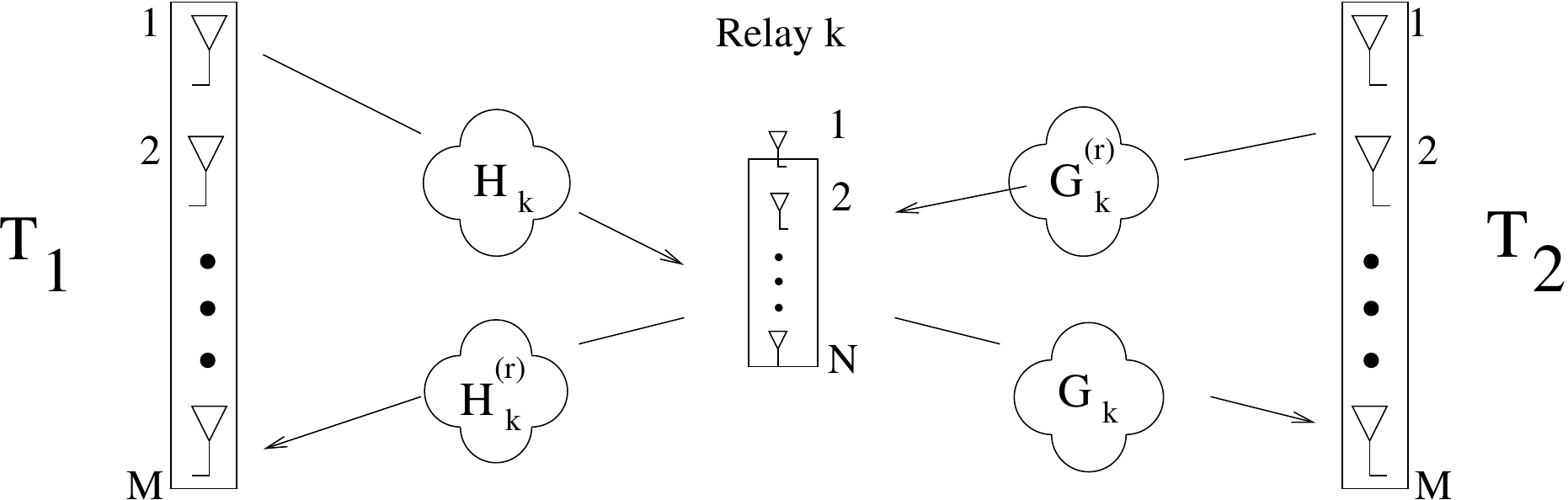}
\caption{Channel Model}
\label{channel}
\end{figure}

\subsection{Channel and Signal Model}
In this paper we assume that all the channels are
frequency flat slow fading block fading channels, where in a block of
time duration $T_c$ (called the coherence time), the channel coefficients
remain constant and change independently from block to block.
We assume that $T_c$ is more that the duration of the time slot used by $T_1$ and $T_2$ to communicate with each other as described before.
As shown in Fig. \ref{channel}, let
the forward channel between $T_1$ and the $k^{th}$
relay be ${\bf H}_k = [{\bf h}_{1k} \ {\bf h}_{2k} \ \ldots \ {\bf h}_{Mk}]$
and the backward channel between $k^{th}$ relay and $T_1$
be ${\bf H}_k^{(r)} = [{\bf h}^{(r)}_{k1} \ {\bf h}^{(r)}_{k2} \ \ldots \ {\bf h}^{(r)}_{kM}]$.
Similarly let the  forward channel between the $k^{th}$
relay and $T_2$ be ${\bf G}_k = [{\bf g}_{k1} \ {\bf g}_{k2} \ \ldots \ {\bf g}_{kM}]$
and the backward channel between $T_2$ and the $k^{th}$
relay be ${\bf G}_k^{(r)} =  [{\bf g}^{(r)}_{1k} \ {\bf g}^{(r)}_{2k}
 \ \ldots \ {\bf g}^{(r)}_{Mk}]$.
We assume that ${\bf H}_k, {\bf G}_k^{(r)} \in {\bbC}^{N\times
M}, {\bf H}_k^{(r)}, {\bf G}_k \in {\mathbb C}^{M\times N}$ with independent and
 identically
distributed (i.i.d.) ${\cal CN}(0,1)$ entries to keep the analysis simple and
tractable.
The ideas presented in this paper, however, apply to a broad class of channel distributions.

In the transmit phase, the $N\times 1$ received signal at the $k^{th}$ relay is given by
\begin{equation}
\label{relayrx}
{\bf r}_k =  \sqrt{\frac{PE_k}{M}} {\bf H}_{k}{\bf x} + \sqrt{\frac{PF_k}{M}}{\bf G}_{k}^{(r)}{\bf u} + {\bf n}_k
\end{equation}
where ${\bf
x}$ and ${\bf u}$ are the $M\times 1$ signals transmitted from $T_1$ and $T_2$  to
be decoded at $T_2$ and $T_1$ respectively, with ${\bbE }\{{\bf x}^{*}{\bf x}\} = {\bbE }\{{\bf u}^{*}{\bf u}\} = M$, $P$ is the power transmitted by $T_1$ and  $T_2$ and $E_k$
and $F_k$ are the path loss and shadowing parameters from $T_1$ and $T_2$ to the $k^{th}$ relay,
respectively. The noise $ {\bf n}_k$ is a
spatio-temporal white complex Gaussian random vector independent across
relays, with ${\bbE}({\bf n}_k{\bf n}_k^*) = \sigma^2 {\bf I}_N$.
Relay $k$ processes its incoming signal to transmit a $N\times 1$ signal
$\sqrt{\gamma_k}{\bf t}_k$ (with ${\bbE }\{{\bf t}_k^{*}{\bf t}_k\} = 1 $)
in the receive phase so that the transmitted power is $\gamma_k$. We assume a power
constraint of $P$ at both $T_1$ and $T_2$ and a sum power constraint
of $P_R $ across all the relays, i.e. ($\sum_{k=1}^K\gamma_k \le P_R$).
The $M \times 1$ received signal
${\bf v}$ and ${\bf y}$ at terminal $T_1$ and
$T_2$ respectively in the receive phase, are given by
\begin{equation}
\label{t1rx}
{\bf v} = \sum_{k=1}^{K}\sqrt{\gamma_kQ_k}{\bf H}_k^{(r)}{\bf t}_k + {\bf w}
\end{equation}
\begin{equation}
\label{t2rx}
{\bf y} = \sum_{k=1}^{K}\sqrt{\gamma_kP_k}{\bf G}_k{\bf t}_k + {\bf z}
\end{equation}
where $\gamma_k$ is the power transmitted by the $k^{th}$ relay, $Q_k, P_k$ are the path loss and shadowing parameters from the $k^{th}$ relay to $T_1$ and $T_2$, respectively, while ${\bf w}$ and ${\bf z}$
are $M\times 1$ spatio-temporal white complex Gaussian noise vectors with
${\bbE}({\bf w}{\bf w}^*) = {\bbE}({\bf z}{\bf z}^*) = \sigma^2 {\bf I}_M$.

The path loss and shadowing effect parameters $E_k, P_k, F_k$ and $Q_k$ $\forall \ k$
for the link between $T_1 \rightarrow T_2$ and $T_2 \rightarrow T_1$, are assumed to be independent
and identically distributed (i.i.d.)
random variables, strictly positive, bounded and remain constant over
the entire time period of interest.

Throughout this paper we assume that both $T_1$ and $T_2$ perfectly know
$\{{\bf H}_k, {\bf H}_k^{(r)}, {\bf G}_k, {\bf G}_k^{(r)}\} \ \forall  \ k,$ $ k=1,2,,\ldots K$
in the receive mode. To be precise, in the receive phase (i.e. when $T_1$ and $T_2$ receive
signal from all the relays),
$T_1$ and $T_2$ both know $\{{\bf H}_k, {\bf G}_k\}$
and $\{{\bf H}_k^{(r)}, {\bf G}_k^{(r)}\}$  $ \forall  \ k, \ k=1,2,,\ldots K$.
We also assume that no transmit CSI is available at $T_1$ and $T_2$,
i.e. in the transmit phase $T_1$ and $T_2$ have
no information about what the realization of ${\bf H}_k$ and $ {\bf G}_k$ is going to be when it transmits its signal to all the relays in the transmit phase, respectively.

In this paper we consider
two different assumptions about the CSI
at the relays. The first case we consider is the {\it coherent} MIMO two-way
relay channel, where
all the relays have CSI in the transmit as
well as the receive phases. For the coherent MIMO two-way relay channel,
in the transmit phase the $k^{th}$ relay knows
the realization of ${\bf H}_k, {\bf G}_k^{(r)}$ and in the receive phase
it knows the realization of ${\bf G}_k, {\bf H}_k^{(r)}$, which could be achieved through channel reciprocity or feedback.
We also consider the {\it non-coherent} MIMO two-way
relay channel where we assume that no CSI is available at
any relay.

\section{Upper Bound on The Capacity Region of The Coherent MIMO Two-Way Relay
Channel
 }
\label{upbound}
The main result in this section is an upper bound on the rate $R_{12}$ and
$R_{21}$ of reliable transmission from $T_1$ to $T_2$ and from $T_2$ to $T_1$,
given by the next Theorem.
\begin{thm}
\label{upperbound}
The capacity region of the
coherent MIMO two-way relay channel is upper bounded by
\[\lim_{K \rightarrow \infty} R_{12} \mapright{w.p. 1} \frac{M}{2}\log{K}+ {\cal O}(1)\]
\[\lim_{K \rightarrow \infty}R_{21} \mapright{w.p. 1} \frac{M}{2}\log{K}+ {\cal O}(1).\]
\end{thm}

\begin{figure}
\centering
\includegraphics[height= 2in]{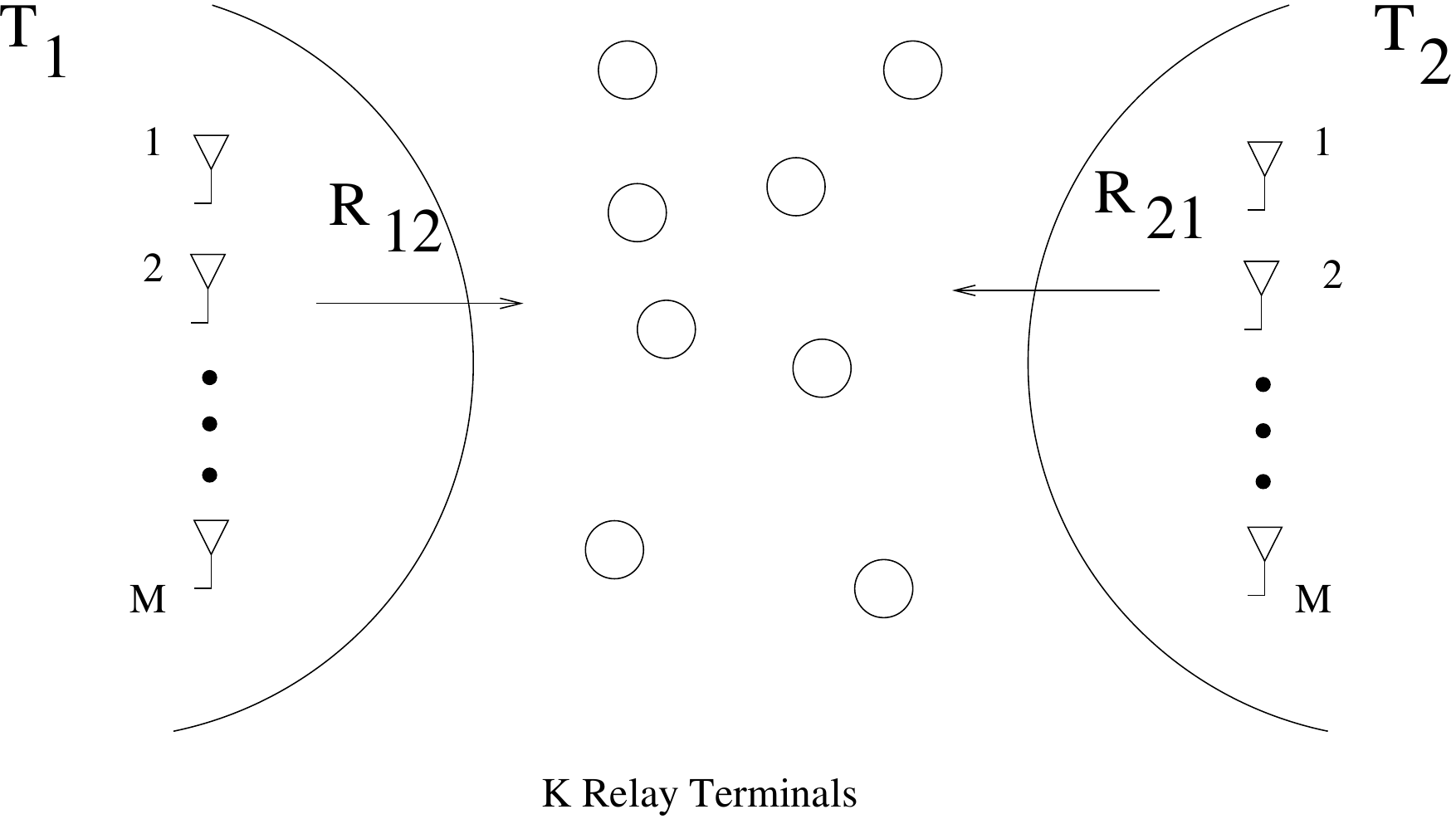}
\caption{Broadcast Cut}
\label{cutset}
\end{figure}

{\it Outline of the Proof:} We start by first separating $T_1$ and
then $T_2$ from the network and apply the cut-set bound \cite{cover}
to upper bound the rate of information transfer between $T_1
\rightarrow T_2$ and $T_2 \rightarrow T_1$, respectively. Using the
cut-set bound, we first show that the maximum rate of information
transfer  from $T_1 \rightarrow T_2$ ($T_2 \rightarrow T_1$) is
upper bounded by the maximum rate of information transfer between
$T_1$ ($T_2$) and relays $1$ to $K$ (the broadcast cut) and also
by the maximum rate of information transfer between relays $1$ to $K$ and $T_2$ ($T_1$) (the multiple access cut), Fig.
\ref{cutset} and Fig. \ref{cutsetmac}. Then we use the capacity result
from Section $4.1$ \cite{tel} to upper bound the maximum rate
through the broadcast cut for the case when CSI is only available at
the receiver (all relays) and all the relays collaborate to decode
the information. Similarly, for the multiple access cut as shown in
Fig. \ref{cutsetmac}, we upper bound the maximum rate at which all
relays can communicate to $T_2$ and $T_1$ by using
the capacity result from Section $3.1$ \cite{tel}, when CSI is known
both at all the relays and at $T_1$ and $T_2$ and all the relays collaborate
to transmit the information.

\begin{rem}
For the broadcast cut, the upper bound on the capacity of
the MIMO one-way relay channel
can be found in \cite{bol} which also trivially serves as an upper bound on
$R_{12}$ and $R_{21}$. It is easy to identify, however, that there is a gap in the proof of Theorem $1$ \cite{bol}. In Theorem $1$ \cite{bol},
it is argued that
\[I\left(\bs; \br_1, \br_2, \ldots, \br_K |
\bt_1, \bt_2, \ldots, \bt_K\right) = I\left(\bs; \br_1, \br_2, \ldots, \br_K\right)\] where $\bt_k$ is a function of $\br_k, \ k=1,2,\ldots, K$, to which
a counterexample can be easily found.
Thus, we do not use result of \cite{bol} directly and attempt a
different proof which is quite similar to the one given in \cite{bol}, but
closes the gap.
\end{rem}

The formal proof is as follows.

\begin{proof}
Throughout this proof we assume that both $T_1$ and $T_2$ perfectly
know ${\bf H}_k, {\bf H}_k^{(r)}, {\bf G}_k, {\bf G}_k^{(r)}$,
$E_k, F_k, Q_k,$ and $P_k$ for $ k=1,2,\ldots K$,
in the receive phase, and the $k^{th}$
relay knows ${\bf H}_k, {\bf
G}_k^{(r)}, E_k, F_k$ in the transmit phase and
${\bf H}_k^{(r)}, {\bf G}_k, Q_k, P_k$ in the receive phase.
For notational simplicity, we do not include ${\bf H}_k, {\bf
G}_k^{(r)}, E_k, F_k,  {\bf H}_k^{(r)}, {\bf G}_k, Q_k, P_k$
in the mutual information expressions. We clearly point out, though,
 whenever their knowledge is used to derive the upper bound.

{\bf Broadcast cut} - To prove the upper bound we make use of the cut-set bound (Section 14.10 \cite{cover}).
Separating the terminal $T_1$ from the rest of the network and applying
the cut-set bound on the broadcast cut as shown in Fig. \ref{cutset},
\begin{equation}
\label{uppbound1}
R_{12} \le I({\bf x};{\bf r}_1, {\bf r}_2,\ldots , {\bf r}_K, {\bf y} | {\bf t}_1, {\bf t}_2, \ldots {\bf t}_K, {\bf u}).
\end{equation}
Applying the cut-set bound while separating the terminal $T_2$,
\begin{equation}
\label{uppbound2}
R_{21} \le I({\bf u};{\bf r}_1,{\bf r}_2,\ldots ,{\bf r}_K, {\bf v} | {\bf t}_1, {\bf t}_2, \ldots {\bf t}_K, {\bf x})
\end{equation}
for some joint distribution $p({\bf x},{\bf t}_1,{\bf t}_2,\ldots,{\bf t}_K, {\bf u})$.
By definition of mutual information \cite{cover}
\begin{eqnarray*}
I({\bf x};{\bf r_1},{\bf r}_2,\ldots ,{\bf r}_K,{\bf y} | {\bf t}_1, {\bf t}_2, \ldots {\bf t}_K, {\bf u})  & = &
I({\bf x};{\bf r}_1,{\bf r}_2,\ldots ,{\bf r}_K | {\bf t}_1, {\bf t}_2, \ldots {\bf t}_K, {\bf u}) \\
& & + \  I({\bf x};{\bf y} |{\bf r}_1,{\bf r}_2,\ldots ,{\bf r}_K, \ {\bf t}_1, {\bf t}_2, \ldots {\bf t}_K, {\bf u}).
\end{eqnarray*}
Expanding the mutual information in terms of differential entropy,
\begin{eqnarray*}
I({\bf x};{\bf r}_1,{\bf r}_2,\ldots ,{\bf r}_K | {\bf t}_1, {\bf t}_2, \ldots {\bf t}_K, {\bf u}) & = & h({\bf x} | {\bf t}_1, {\bf t}_2, \ldots {\bf t}_K, {\bf u}) \\
& & - \ h({\bf x}| {\bf r}_1,{\bf r}_2,\ldots ,{\bf r}_K, {\bf t}_1, {\bf t}_2, \ldots {\bf t}_K, {\bf u} ).
\end{eqnarray*}
Since conditioning can only reduce entropy \cite{cover},
\begin{eqnarray*}
I({\bf x};{\bf r}_1,{\bf r}_2,\ldots ,{\bf r}_K | {\bf t}_1, {\bf t}_2, \ldots {\bf t}_K, {\bf u}) & \le & h({\bf x}|{\bf u}) \\
& & - \ h({\bf x}| {\bf r}_1,{\bf r}_2,\ldots ,{\bf r}_K, {\bf t}_1, {\bf t}_2, \ldots {\bf t}_K, {\bf u} ).
\end{eqnarray*}
Note that ${\bf t}_1,{\bf t}_2,\ldots ,{\bf t}_K$ is a function of
${\bf r}_1,{\bf r}_2,\ldots ,{\bf r}_K$, which implies
\begin{eqnarray*}
I({\bf x};{\bf r}_1,{\bf r}_2,\ldots ,{\bf r}_K | {\bf t}_1, {\bf t}_2, \ldots {\bf t}_K, {\bf u}) & \le & h({\bf x} |{\bf u}) \\
& & - \ h({\bf x}| {\bf r}_1,{\bf r}_2,\ldots ,{\bf r}_K, {\bf u} )
\end{eqnarray*}
and hence
\[I({\bf x};{\bf r}_1,{\bf r}_2,\ldots ,{\bf r}_K | {\bf t}_1, {\bf t}_2, \ldots {\bf t}_K, {\bf u}) \le I({\bf x};{\bf r}_1,{\bf r}_2,\ldots ,{\bf r}_K | {\bf u}).\]
From (\ref{t2rx}), with knowledge of $P_k$ and ${\bf G}_k,\ \forall \ k$ at terminal $T_2$,
\[I({\bf x};{\bf y} |{\bf r}_1,{\bf r}_2,\ldots ,{\bf r}_K, \  {\bf t}_1, {\bf t}_2, \ldots {\bf t}_K, {\bf u}) = I({\bf x},{\bf z})\]
where ${\bf z}$ is the AWGN noise. Since ${\bf x}$ and ${\bf z}$
are independent, $I({\bf x},{\bf z}) = 0$, and therefore
\[I({\bf x};{\bf r_1},{\bf r}_2,\ldots ,{\bf r}_K,{\bf y} | {\bf t}_1, {\bf t}_2, \ldots {\bf t}_K, {\bf u}) \le I({\bf x};{\bf r}_1,{\bf r}_2,\ldots ,{\bf r}_K | {\bf u}).\]
Note that
\begin{eqnarray}
\label{mieq} \nonumber
I({\bf x};{\bf r}_1,{\bf r}_2,\ldots ,{\bf r}_K,| {\bf u})  &=&
I\left({\bf x};\frac{{\bf r}_1}{\sqrt{K}},\frac{{\bf r}_2}{\sqrt{K}},\ldots ,\frac{{\bf r}_K}{\sqrt{K}}| {\bf u}\right) \\
 &=& h\left(\frac{{\bf r}_1}{\sqrt{K}},\frac{{\bf r}_2}{\sqrt{K}},\ldots ,\frac{{\bf r}_K}{\sqrt{K}}|{\bf u}\right) - h\left(\frac{{\bf r}_1}{\sqrt{K}},\frac{{\bf r}_2}{\sqrt{K}},\ldots ,\frac{{\bf r}_K}{\sqrt{K}}|{\bf x , u}\right).
\end{eqnarray}

Next we evaluate (\ref{mieq}) using (\ref{relayrx}). From (\ref{relayrx}),
\[{\bf r}_k = \sqrt{\frac{PE_k}{M}}{\bf H}_k{\bf x} +  \sqrt{\frac{PF_k}{M}}{\bf G}^{(r)}_k{\bf u} + {\bf n}_k.\]
Now with knowledge of $F_k$ and ${\bf G}_k^{(r)}$ at each relay
\[h\left(\frac{{\bf r}_1}{\sqrt{K}},\frac{{\bf r}_2}{\sqrt{K}},\ldots ,\frac{{\bf r}_K}{\sqrt{K}}|{\bf u}\right) =
h\left(\sqrt{\frac{PE_1}{KM}} {\bf H}_1{\bf x} + \frac{{\bf n}_1}{\sqrt{K}},
\sqrt{\frac{PE_2}{KM}} {\bf H}_2{\bf x} + \frac{{\bf n}_2}{\sqrt{K}},
\ldots,
\sqrt{\frac{PE_K}{KM}}{\bf H}_K{\bf x} + \frac{{\bf n}_K}{\sqrt{K}}|{\bf u}\right).\]
Since conditioning can only decrease entropy,
\[h\left(\frac{{\bf r}_1}{\sqrt{K}},\frac{{\bf r}_2}{\sqrt{K}},\ldots ,\frac{{\bf r}_K}{\sqrt{K}}|{\bf u}\right) \le h\left(\sqrt{\frac{PE_1}{KM}} {\bf H}_1{\bf x} + \frac{{\bf n}_1}{\sqrt{K}},
\sqrt{\frac{PE_2}{KM}} {\bf H}_2{\bf x} + \frac{{\bf n}_2}{\sqrt{K}},
\ldots,
\sqrt{\frac{PE_K}{KM}}{\bf H}_K{\bf x} + \frac{{\bf n}_K}{\sqrt{K}}\right).\]
With perfect knowledge of $E_k, F_k$ and ${\bf H}_k, {\bf G}_k^{(r)}$ at
each relay
\[h\left(\frac{{\bf r}_1}{\sqrt{K}},\frac{{\bf r}_2}{\sqrt{K}},\ldots ,\frac{{\bf r}_K}{\sqrt{K}}|{\bf x , u}\right)
= h\left(\frac{{\bf n}_1}{\sqrt{K}}, \frac{{\bf n}_2}{\sqrt{K}},\ldots , \frac{{\bf n}_K}{\sqrt{K}}\right)\]
and using (\ref{mieq}) it follows that
\begin{eqnarray*}
I({\bf x};{\bf r}_1,{\bf r}_2,\ldots ,{\bf r}_K|{\bf u}) & \le &
h\left(\sqrt{\frac{PE_1}{KM}}{\bf H}_1{\bf x} + \frac{{\bf n}_1}{\sqrt{K}},
\sqrt{\frac{PE_2}{KM}}{\bf H}_2{\bf x} + \frac{{\bf n}_2}{\sqrt{K}},
\ldots,
\sqrt{\frac{PE_K}{KM}}{\bf H}_K{\bf x} + \frac{{\bf n}_K}{\sqrt{K}}\right)\\
& & -  h\left(\frac{{\bf n}_1}{\sqrt{K}},\frac{{\bf n}_2}{\sqrt{K}},\ldots ,\frac{{\bf n}_K}{\sqrt{K}}\right).
\end{eqnarray*}
Using the capacity result from Section $4.1$ \cite{tel}
 when CSI is only known at the receiver, the R.H.S. can be upper bounded by
\[\log\det\left({\bf I}_M + \frac{1}{\frac{\sigma^2}{K}}\sum_{k=1}^{K}\frac{PE_k}{KM}{\bf H}_k^*{\bf H}_k\right)\]
and the maximum is achieved when ${\bf x}$ is circularly symmetric
complex Gaussian with ${\bbE}({\bf x}{\bf x}^*)= {\bf I}_M$. Thus, it follows that
\begin{equation}
\label{alphaupbound1}
I({\bf x};{\bf r}_1,{\bf r}_2,\ldots ,{\bf r}_K|{\bf u})\le \log\det\left({\bf I}_M + \frac{1}{\frac{\sigma^2}{K}}\sum_{k=1}^{K}\frac{PE_k}{KM}{\bf H}_k^*{\bf H}_k\right).
\end{equation}

Similarly, by interchanging the roles of ${\bf x}$ and ${\bf u}$ and replacing
$E_k$ with $F_k$ and ${\bf H}_k$ with ${\bf G}_k$,
\begin{equation}
\label{alphaupbound2}
I({\bf u};{\bf r}_1,{\bf r}_2,\ldots ,{\bf r}_K|{\bf x}) \le \log\det\left({\bf I}_M + \frac{1}{\frac{\sigma^2}{K}}\sum_{k=1}^{K}\frac{PF_k}{KM}{\bf G}_k^{(r)*}{\bf G}_k^{(r)}\right).
\end{equation}

Using the strong law of large numbers
\[\lim_{K \rightarrow \infty}\frac{1}{K}\sum_{k=1}^{K}\frac{PE_k}{M}{\bf H}_k^*{\bf H}_k \xrightarrow{w.p.1}\frac{P}{M}\bbE\left\{E_k{\bf H}_k^*{\bf H}_k\right\}\] and
\[\lim_{K \rightarrow \infty}\frac{1}{K}\sum_{k=1}^{K}\frac{PF_k}{M}{\bf G}_k^{(r)*}{\bf G}_k^{(r)}
 \xrightarrow{w.p.1}\frac{P}{M}\bbE\left\{F_k{\bf G}_k^{(r)*}{\bf G}_k^{(r)}\right\}.\]
Since
${\bbE}\{{\bf H}_k^*{\bf H}_k\} = {\bbE}\{{\bf G}_k^{(r)*}{\bf G}^{(r)}_k\} = N{\bf I}_M$ and let $\bbE\{E_k\} = \bbE\{F_k\} = \mu$, using
(\ref{uppbound1}), (\ref{uppbound2}), (\ref{alphaupbound1}), (\ref{alphaupbound2}) and the fact
that the sources $ T_1$ and $T_2$ transmit only for
$\alpha$ fraction of the time in each time slot, it follows that

\begin{equation}
\label{12upper}
\lim_{K \rightarrow \infty}R_{12} \mapright{w.p. 1} \alpha M\log\left(1+ \frac{KNP\mu}{M\sigma^2}\right)
\end{equation}
and
\begin{equation}
\label{21upper}
\lim_{K \rightarrow \infty}R_{21} \mapright{w.p. 1} \alpha M\log\left(1+ \frac{KNP\mu}{M\sigma^2}\right).
\end{equation}


Since $M, N, P, \mu$ and $\sigma^2$ are finite integers, as $K \rightarrow \infty$
\begin{equation}
\label{finalupboundbc1}
\lim_{K \rightarrow \infty}R_{12} \mapright{w.p. 1} \alpha M\log(K) + {\cal O}(1)
\end{equation}
and
\begin{equation}
\label{finalupboundbc2}
\lim_{K \rightarrow \infty}R_{21} \mapright{w.p. 1} \alpha M\log(K) + {\cal O}(1).
\end{equation}

{\bf Multiple access cut} - Again by using the cut-set bound, we bound the
maximum rate of information transfer $R_{12}$ ($R_{21}$) from $T_1
\rightarrow T_2$ ($T_2 \rightarrow T_1$) by the maximum rate of
information transfer across the multiple access cut as shown in Fig.
\ref{cutsetmac}. Using the cut-set bound, $R_{12}$ and $R_{21}$ are bounded by
\begin{figure}
\centering
\includegraphics[height= 2in]{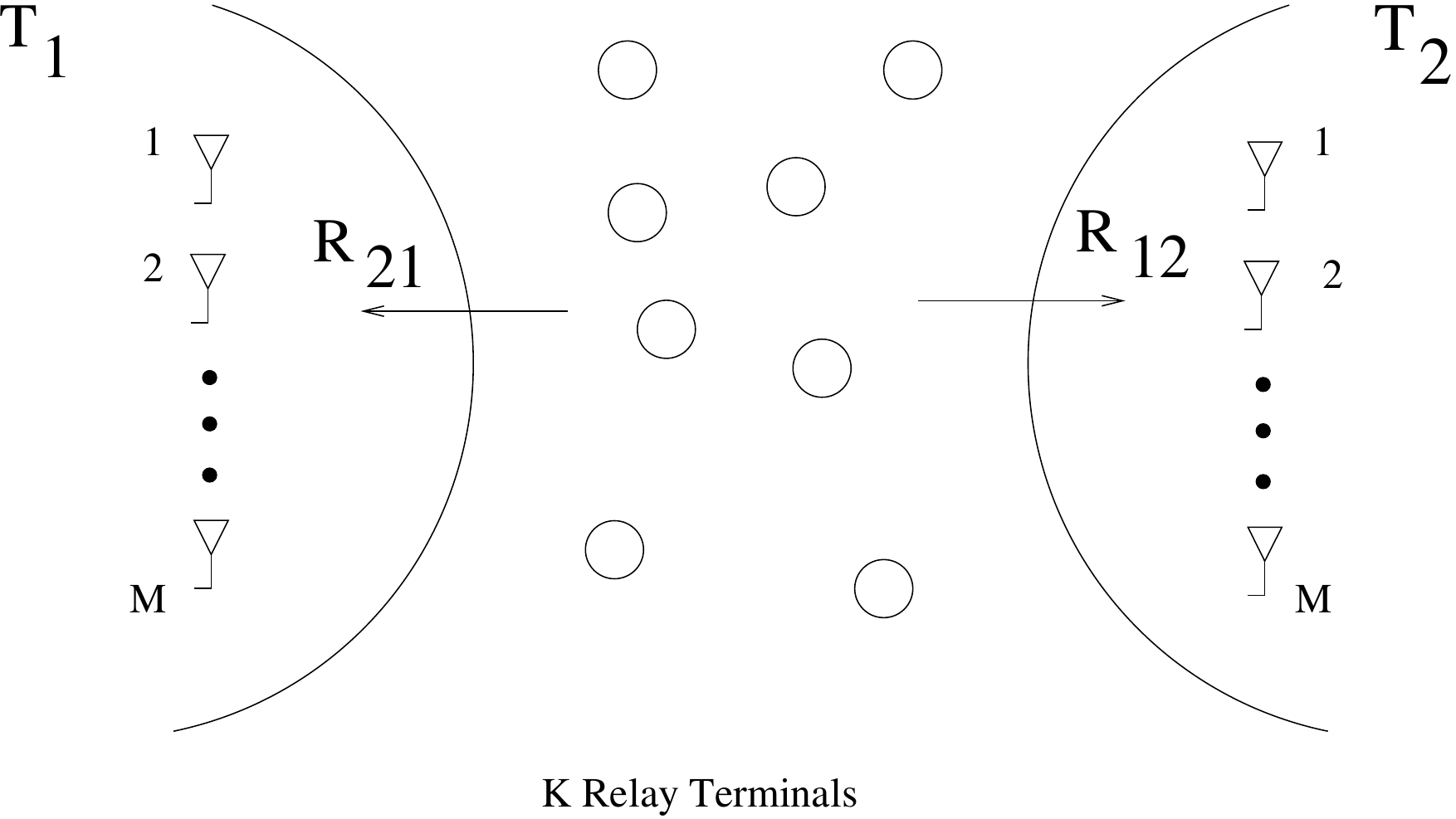}
\caption{Multiple Access Cut}
\label{cutsetmac}
\end{figure}
\begin{equation}
\label{upboundmac1}
R_{12} \le I({\bf x},  {\bf t}_1,  {\bf t}_2,  \ldots, {\bf t}_K ; {\bf y} | \bf {u})
\end{equation}
\begin{equation}
\label{upboundmac2}
R_{21} \le I({\bf u}, {\bf t}_1, {\bf t}_2, \ldots,  {\bf t}_K ; {\bf v} | \bf {x}).
\end{equation}
By definition of mutual information
\begin{eqnarray*}
I({\bf x}, {\bf t}_1, {\bf t}_2, \ldots,  {\bf t}_K ; {\bf y} | \bf {u})
&=&h({\bf y} | {\bf u}) - h({\bf y}| {\bf t}_1, {\bf t}_2, \ldots, {\bf t}_K, \bf {u}) \\
& & + \ h({\bf y}| {\bf t}_1, {\bf t}_2, \ldots, {\bf t}_K, {\bf u})
- h({\bf y}| {\bf t}_1, {\bf t}_2, \ldots, {\bf t}_K,  {\bf x}, \bf {u}).
\end{eqnarray*}
Note that given ${\bf t}_1, {\bf t}_2, \ldots, {\bf t}_K$, ${\bf y}$ is
independent of ${\bf x}$ and ${\bf u}$, thus
\[h({\bf y}| {\bf t}_1, {\bf t}_2, \ldots, {\bf t}_K,  {\bf x}, {\bf u}) =
 h({\bf y}| {\bf t}_1, {\bf t}_2, \ldots, {\bf t}_K, {\bf u}) = h({\bf y}| {\bf t}_1, {\bf t}_2, \ldots, {\bf t}_K).\]
Therefore
\[I({\bf x}, {\bf t}_1, {\bf t}_2, \ldots,  {\bf t}_K ; {\bf y} | {\bf u})
= h({\bf y} |{\bf u}) -  h({\bf y}| {\bf t}_1, {\bf t}_2, \ldots, {\bf t}_K). \]
Since conditioning can only reduce entropy,
\begin{eqnarray*}
I({\bf x}, {\bf t}_1, {\bf t}_2, \ldots,  {\bf t}_K ; {\bf y} | {\bf u})
&\le&  h({\bf y}) -  h({\bf y}| {\bf t}_1, {\bf t}_2, \ldots, {\bf t}_K) \\
&= & I({\bf t}_1, {\bf t}_2, \ldots, {\bf t}_K; {\bf y}).
\end{eqnarray*}
Hence from (\ref{upboundmac1}),
\begin{equation}
\label{t1t2mac}
R_{12} \le I({\bf t}_1, {\bf t}_2, \ldots, {\bf t}_K ; {\bf y}).
\end{equation}
Following similar steps
\begin{equation}
\label{t2t1mac}
R_{21} \le I({\bf t}_1, {\bf t}_2, \ldots, {\bf t}_K ; {\bf v}).
\end{equation}
Clearly $R_{12}, R_{21}$ are bounded by the maximum rate of
information across the multiple access cut Fig. \ref{cutsetmac}.

Next, we compute $I({\bf t}_1, {\bf t}_2, \ldots, {\bf t}_K ; {\bf y})$.
Recall From (\ref {t2rx}), that the received signal ${\bf y}$ at $T_2$ is
\[{\bf y} = \sum_{k=1}^{K}\sqrt{\gamma_kP_k}{\bf G}_k{\bf t}_k + {\bf z}.\]
Note that \[I({\bf t}_1, {\bf t}_2, \ldots, {\bf t}_K; {\bf y}) =
I\left({\bf t}_1, {\bf t}_2, \ldots, {\bf t}_K; \frac{{\bf y}}{\sqrt{K}}\right).\]
Dividing ${\bf y}$ by $\sqrt{K}$, the scaled signal is
\[\frac{{\bf y}}{\sqrt{K}} = \frac{1}{\sqrt{K}}\sum_{k=1}^{K}\sqrt{\gamma_kP_k}{\bf G}_k{\bf t}_k + \frac{{\bf z}}{\sqrt{K}}.\]
This can also be written as
\[\frac{{\bf y}}{\sqrt{K}} = \underbrace{\frac{1}{\sqrt{K}}
\left[\sqrt{P_1}{\bf G}_1 \ \sqrt{P_2}{\bf G}_2 \ \ldots \
\sqrt{P_K}{\bf G}_K\right]}_{\Phi}
\left[\sqrt{\gamma_1}{\bf t}_1 \sqrt{\gamma_2}{\bf t}_2 \ldots \sqrt{\gamma_K}{\bf t}_K\right]^T + \frac{{\bf z}}{\sqrt{K}}.\]
Note that $\Phi$ is a $M \times NK$ matrix.
Now assuming that all the relays know ${\bf G}_k \ \forall
 k$ (allowing cooperation among all relays), with total power available across all relays bounded by $P_R$, from Section $3.1$ \cite{tel},
\begin{equation}
\label{waterfill}
I\left({\bf t}_1, {\bf t}_2, \ldots, {\bf t}_K ;\frac{{\bf y}}{\sqrt{K}}\right) \le \sum_{l=1}^{\min{\{NK,M\}}}
\max{\left\{0,\log\left(\frac{K\lambda_l\nu}{\sigma^2}\right)\right\}}
\end{equation}
where $\lambda_l, l=1,2, \ldots, \min{\{NK,M\}}$ are the eigen values of $\Phi\Phi^*$ matrix  and $\nu$ is chosen such that
\[ \sum_{l=1}^{\min{\{NK,M\}}} \max\left\{0, \nu - \frac{1}{\lambda_l}\right\} = P_R. \]
By definition,
$\Phi\Phi^* = \frac{1}{K}\sum_{k=1}^K P_k\bG_k\bG_k^*$.
From the strong law of large numbers
\[\lim_{K \rightarrow \infty}\frac{1}{K}\sum_{k=1}^KP_k\bG_k\bG_k^* \xrightarrow{w.p. 1}
{\bbE}\left\{P_k\bG_k\bG_k^*\right\} =
{\bbE}\left\{P_k\right\}{\bbE}\left\{\bG_k\bG_k^*\right\} = \mu N{\bf I}_M \]
since ${\bbE}\left\{\bG_k\bG_k^*\right\} = N{\bf I}_M$ and
$\mu \bydef {\bbE}\left\{P_k\right\}$.
Therefore, it follows that
\[\lambda_i = N\mu  \ \ \forall \ i = 1,2, \ldots M.\]
which implies
\[\nu = \left(\frac{P_R}{M} + \frac{1}{N\mu}\right)\] and
\[I\left({\bf t}_1, {\bf t}_2, \ldots, {\bf t}_K; \frac{{\bf y}}{\sqrt{K}}\right) \mapright{w.p. 1} \sum_{l=1}^{M} \log\left(\frac{KN\rho}{\sigma^2}\left(\frac{P_R}{M} + \frac{1}{N\mu}\right)\right).\]
Since $M, N, P_R, \sigma^2$ and $\mu$ are all finite,
\[\lim_{K\rightarrow \infty}I\left({\bf t}_1, {\bf t}_2, \ldots, {\bf t}_K; \frac{{\bf y}}{\sqrt{K}}\right) \mapright{w.p. 1} M \log{K} + {\cal O}(1).\]
Moreover, since the relays transmit only for $(1-\alpha)$ fraction of
time in any given time slot,
\begin{equation}
\label{finalupboundmac1}
\lim_{K \rightarrow \infty} R_{12} \mapright{w.p. 1} (1-\alpha)
I\left({\bf t}_1, {\bf t}_2, \ldots, {\bf t}_K; \frac{{\bf y}}{\sqrt{K}}\right)
 \le (1-\alpha)M \log{K} + {\cal O}(1).
\end{equation}
Similarly, we can derive a bound for $R_{21}$ by using (\ref{t1rx}) and
(\ref{t2t1mac}),
\begin{equation}
\label{finalupboundmac2}
\lim_{K \rightarrow \infty} R_{21} \mapright{w.p. 1} (1-\alpha)
I\left({\bf t}_1, {\bf t}_2, \ldots, {\bf t}_K; \frac{{\bf v}}{\sqrt{K}}\right)
 \le (1-\alpha)M \log{K} + {\cal O}(1).
\end{equation}
Combining (\ref{finalupboundbc1}), (\ref{finalupboundbc2}),
(\ref{finalupboundmac1}) and (\ref{finalupboundmac2})
\[\lim_{K \rightarrow \infty}R_{12} \mapright{w.p. 1} \min{\{\alpha, 1-\alpha\}}M\log{K}+ {\cal O}(1)\]
\[\lim_{K \rightarrow \infty}R_{21} \mapright{w.p. 1} \min{\{\alpha, 1-\alpha\}}M\log{K}+ {\cal O}(1).\]
Since $\alpha \in [0,1]$, $\min{\{\alpha, 1-\alpha\}} \le \frac{1}{2}$, therefore
\[\lim_{K \rightarrow \infty}R_{12} \mapright{w.p. 1} \frac{M}{2}\log{K}+ {\cal O}(1)\]
\[\lim_{K \rightarrow \infty}R_{21} \mapright{w.p. 1} \frac{M}{2}\log{K}+ {\cal O}(1).\]

\end{proof}
{\it Discussion: }
In Theorem \ref{upperbound}, we obtained upper bounds on $R_{12}$ and
$R_{21}$ by using cut-set bound on the broadcast cut (Fig. \ref{cutset}) and the multiple access cut (Fig. \ref{cutsetmac}). For the broadcast cut, the upper bound corresponds to the case when the transmitter $T_1$ or $T_2$ has no CSI while all the relays collaborate to decode the message sent by $T_1$ or $T_2$ with perfect CSI, while the upper bound in the multiple access cut corresponds to the case when all the relays collaborate to transmit data to $T_1$ or $T_2$ using all their $NK$ antennas with transmit CSI available at all relays.
An important point to note is that the upper
bound obtained in Theorem \ref{upperbound} is
for any arbitrary $\alpha$, which implies that the upper bound
is valid for all two-phased MIMO two-way relay channel protocols and not for only
$\alpha = 1/2$ as is the case in \cite{bol}.


In the next section we illustrate a simple amplify and forward (AF) strategy
 whose achievable rate is a  constant term away from the upper bound.

\section{Lower Bound on The Capacity Region of The Coherent MIMO Two-Way Relay Channel}
\label{ach} In this section we propose an AF strategy to achieve the
upper bound obtained in Theorem \ref{upperbound} on the capacity
region of the MIMO two-way relay channel within a constant term.
The motivation to consider AF is because with DF, at
each relay, to decode $T_1$'s message $T_2$'s message
is treated as interference and vice-versa, which implies that
the achievable rate region with DF is same as that of the
achievable rate region
for the multiple access channel \cite{cover}.
Since the achievable rate region of the multiple access channel is
strictly less than the upper
bound derived in Theorem \ref{upperbound} one cannot hope
to achieve the upper bound given by Theorem \ref{upperbound}
using DF protocol. With AF, however, each relay processes the
received signal using its CSI and transmits it to $T_1$ and $T_2$ in
the receive phase without any decoding. Since both $T_1$ and $T_2$
know what they transmit, (i.e. $T_1$ knows ${\bf x}$ and $T_2$ knows
${\bf u}$) with perfect receive CSI, both $T_1$ and $T_2$ can cancel
the contribution of their own transmitted signal from the received
signal and decode other terminal's message without any self
interference.

Before discussing the MIMO two-way relay channel with multiple relays,
let us first consider the case of a MIMO one-way relay channel (i.e. $T_2$ has no
data for $T_1$) with only one relay.
For this case, the optimal AF strategy to maximize
mutual information at the destination is to multiply $\bV_2\bD\bU_1^{*}$ to the signal at the relay,
where the singular value decomposition of
${\bf H}_1$ is $ \bU_1\bD_1\bV_1^*$ and ${\bf G}_1$ is  $\bU_2\bD_2\bV_2^*$ and
$\bD$ is a diagonal matrix whose entries are chosen by waterfilling  \cite{medina}.
Finding the optimal AF strategy for the
multiple relay case is a non-trivial problem and has not been found
to the best of our knowledge. Moreover, the two-way nature of our problem
makes it even more difficult to find the optimal AF strategy.

To obtain a lower bound on the achievable rates for the MIMO two-way relay channel
we propose a {\it dual channel matching} AF strategy in which relay $k$ multiplies
$\frac{1}{\sqrt{\beta_k}} \left(\bG_k^*\bH_k^* + {\bf H}_k^{(r)*}{\bf G}^{(r)*}_{k}\right)$ to the received signal and
forwards it to $T_1$ and $T_2$,  where $\beta_k$ is the normalization constant
to satisfy the power constraint.
In this AF strategy, each relay tries to match both the channels
which the data streams from $T_1$ to $T_2$ and $T_2$ to $T_1$ experience.
In dual channel matching,
the complex conjugates
of the channels are used directly rather than the unitary matrices from the
SVD of the channels \cite{medina}. This modification makes it easier to
compute the
achievable rates for the MIMO two-way relay channel.

Together with dual channel matching we
restrict the signal transmitted from $T_1$ and $T_2$, $\bx$ and
$\bu$, respectively, to be circularly symmetric complex Gaussian with
covariance matrix ${\bbE}\{{\bf x}{\bf x}^*\} =
{\bbE}\{{\bf u}{\bf u}^*\} = \bQ$ with $\tr\left( \bQ\right) =M$
(to meet the power constraint) to obtain an
achievable rate region for the coherent MIMO two-way relay channel.
Moreover, we use $\alpha = \frac{1}{2}$ i.e. $T_1$ and $T_2$ transmit and
receive for same amount of time. The achievable rates $R_{12}$ and $R_{21}$ using the
above described AF strategy
are given by the following Theorem.

\begin{thm}
\label{lowerbound}
For the coherent MIMO two-way relay channel, the achievable rates are given by
\begin{eqnarray*}
\lim_{K \rightarrow \infty}R_{12} & \mapequal{w.p.1} & \frac{M}{2}\log{(K)} + {\cal O}(1) \\
\lim_{K \rightarrow \infty}R_{21} & \mapequal{w.p.1} & \frac{M}{2}\log{(K)} + {\cal O}(1)
\end{eqnarray*}
with no cooperation required between $T_1$ and $T_2$.
\end{thm}
\begin{proof}
From (\ref{relayrx}), the received signal at the $k^{th}$ relay is given by
\begin{equation}
\label{rxdcm}
{\bf r}_k =  \sqrt{\frac{PE_k}{M}} {\bf H}_{k}{\bf x} +
\sqrt{\frac{PF_k}{M}}{\bf G}^{(r)}_{k}{\bf u} + {\bf n}_k.
\end{equation}
Using dual channel matching, at relay $k$
the transmitted signal ${\bf t}_k$ is
given by
\begin{equation}
\label{txdcm}
{\bf t}_k =
\frac{\left(\bG_k^*\bH_k^* + {\bf H}_k^{(r)*}{\bf G}^{(r)*}_{k}\right)}{\sqrt{\beta_k}}\br_k
\end{equation}
where $\beta_k$ is to ensure that $\bbE\left\{{\bf t}_k^*{\bf t}_k\right\} =1$.
The received signal at $T_2$ is given by
\begin{equation}
\label{destrecsig}
\by = \sum_{k=1}^K\sqrt{\gamma_kP_k}\bG_k\bt_k + \bz.
\end{equation}
Expanding (\ref{destrecsig}) using (\ref{rxdcm}) and (\ref{txdcm})
\begin{eqnarray*}
\by &=& \underbrace{\sum_{k=1}^K\sqrt{\frac{\gamma_kP_kPE_k}{M\beta_k}}\bG_k
\left(\bG_k^*\bH_k^* + {\bf H}_k^{(r)*}{\bf G}^{(r)*}_{k}\right)\bH_k}_{\bA}\bx +
\sum_{k=1}^K\sqrt{\frac{\gamma_kP_kPF_k}{M\beta_k}}\bG_k
\left(\bG_k^*\bH_k^* + {\bf H}_k^{(r)*}{\bf G}^{(r)*}_{k}\right){\bf G}^{(r)}_{k}\bu\\
&&+
\sum_{k=1}^K\underbrace{\sqrt{\frac{\gamma_kP_k}{\beta_k}}\bG_k
\left(\bG_k^*\bH_k^* + {\bf H}_k^{(r)*}{\bf G}^{(r)*}_{k}\right)}_{\bB_k}{\bf n}_k
+ \bz.
\end{eqnarray*}
Since $\bu$ and the channel coefficients $\bH_k, \bG_k, {\bf H}_k^{(r)}, {\bf G}^{(r)}_{k}$ are known at $T_2,\  \forall \ k$, the second term can be removed from the received signal at $T_2$.
Moreover, as described before $\bx$ is circularly symmetric
complex Gaussian vector
with covariance matrix $\bQ$, thus the achievable rate for
$T_1$ to $T_2$ link is
\cite{tel}
\[R_{12} = \frac{1}{2}\log\det\left(\bI_M +  \frac{\bA\bQ\bA^*}
{\sum_{k=1}^K\bB_k\bB_k^* + \bI_M}\right),\]
since $\bbE\left\{\bn_k\bn_k^*\right\} = \bbE\left\{\bz\bz^*\right\}= \bI_M, \ \forall \ k.$
Using $\bQ = \frac{1}{M}\bI$ and dividing the numerator and denominator
by $K^2$,
\[R_{12} = \frac{1}{2}\log\det\left(\bI_M +
\frac{K}{M}\frac{\frac{\bA}{K}\frac{\bA^*}{K}}
{\frac{1}{K}\sum_{k=1}^K\bB_k\bB_k^* + \frac{1}{K}\bI_M}\right).\]
Note that as $K \rightarrow \infty$, the contribution from $\frac{1}{K}\bI_M$
can be neglected and it follows that
\[R_{12} = \frac{1}{2}\log\det\left(\bI_M +
\frac{K}{M}\frac{\frac{\bA}{K}\frac{\bA^*}{K}}
{\frac{1}{K}\sum_{k=1}^K\bB_k\bB_k^* }\right).\]
Using equal power allocation $\gamma_k = \frac{P_R}{K}$
to satisfy the total
power constraint of $P_R$ across all relays
\[I(\bx;\by) = \log\det\left(\bI_M +
\frac{K}{M}\frac{\frac{\hat{\bA}}{K}\frac{\hat{\bA}^*}{K}}
{\frac{1}{K}\sum_{k=1}^K\hat{\bB}_k\hat{\bB}_k^* }\right)\]
where $\hat{\bA} = \sum_{k=1}^K\sqrt{\frac{P_kPE_k}{M\beta_k}}\bG_k
\left(\bG_k^*\bH_k^* + {\bf H}_k^{(r)*}{\bf G}^{(r)*}_{k}\right)\bH_k$
and $\hat{\bB}_k = \sqrt{\frac{P_k}{\beta_k}}\bG_k
\left(\bG_k^*\bH_k^* + {\bf H}_k^{(r)*}{\bf G}^{(r)*}_{k}\right)$.
Using the strong law of large numbers,
\[\lim_{K\rightarrow\infty}\frac{\hat{\bA}}{K} \xrightarrow{w.p.1} \sqrt{\frac{P}{M}}\kappa M^2\bI_M\] and
 \[\lim_{K\rightarrow\infty}\frac{\hat{\bA^*}}{K} \xrightarrow{w.p.1} \sqrt{\frac{P}{M}}\kappa M^2\bI_M\]
since $\bbE\left\{\bG_k\bG_k^*\right\} = \bbE\left\{\bH_k^*\bH_k\right\} = M\bI_M \ \forall \ k$,
$\bbE\left\{\bG_k\bG_k^{(r)*}\right\}  = \bbE\left\{\bH_k^{(r)*}\bH_k\right\}  = 0\bI_M\ \forall \ k$
and $P_k, E_k, \beta_k$ are i.i.d. with
$\kappa = \bbE\left\{\sqrt{\frac{P_kE_k}{\beta_k}}\right\}.$
Moreover, from the strong law of large numbers
\[\lim_{K\rightarrow\infty}\frac{\sum_{k=1}^K\hat{\bB}_k\hat{\bB}_k^*}{K} \xrightarrow{w.p.1}
\bbE\left\{\hat{\bB}_k\hat{\bB}_k^*\right\} = \theta\bI_M\]
for some finite $\theta$,
since $\hat{\bB}_k\hat{\bB}_k^*$ are i.i.d. for each $k$ and each entry of
$\hat{\bB}_k$ has finite variance.
Thus, using these approximations,
\[\lim_{K \rightarrow \infty}R_{12} \mapequal{w.p.1} \ \frac{1}{2}\log\det\left(\bI_M + \frac{KP\kappa^2 M^2}{\theta}\bI_M\right).\]
\[\lim_{K \rightarrow \infty}R_{21} \mapequal{w.p.1} \frac{M}{2}\log\left(1 + \frac{KP\kappa^2 M^2}{\theta}\right).\]
Since $M,P, \kappa$ and $\theta$ are finite, as $K \rightarrow \infty$,
\[\lim_{K \rightarrow \infty}R_{21} \mapequal{w.p.1} \frac{M}{2}\log K + {\cal O}(1), \]
and similarly
\[\lim_{K \rightarrow \infty}R_{21} \mapequal{w.p.1} \frac{M}{2}\log K + {\cal O}(1).\]
\end{proof}
\vspace{0.25in} {\it Discussion:}
Theorem \ref{lowerbound} shows
that the achievable rate in each direction $T_1 \rightarrow T_2 $ or
$T_2 \rightarrow T_1$ with the coherent MIMO two-way relay channel using
dual channel matching at each relay is given by
$\frac{M}{2}\log{(K)} + {\cal O}(1)$ as
$K \rightarrow \infty$.
More importantly, Theorem \ref{lowerbound} also shows that
both $T_1$ and $T_2$ can simultaneously transmit at rate $\frac{M}{2}\log{(K)} + {\cal O}(1)$, without
affecting each other's data rate and without requiring any
cooperation between themselves.
The result can be interpreted as follows.
With dual channel matching, the transmitted signals
from $T_1$ and $T_2$ are coherently added by all relays and the
equivalent array gain of $\log K$ is obtained at both the
receivers $T_1$ and $T_2$ for all the $M$ data streams transmitted from
$T_2$ and $T_1$. Moreover, with perfect channel knowledge, $T_1$ and $T_2$
can cancel the self interference their own transmitted signals create,
which enable $T_1$ and $T_2$ to simultaneously achieve the rate of
$\frac{M}{2}\log{(K)} + {\cal O}(1)$,
without requiring any cooperation.

Recall that for the MIMO one-way relay channel an AF strategy was proposed in
\cite{bol} to achieve the upper bound within a constant term. With
the AF strategy of \cite{bol}, $M$ independent data streams are
transmitted from $T_1$ and all the relays are divided into $M$ sets
with each set helping a particular data stream and independent
decoding of data streams is employed at the receiver. Compared to
the AF strategy of \cite{bol}, with dual channel matching all relays
participate in transmission of all data streams from $T_1$ to $T_2$
and $T_2$ to $T_1$ and thus provides a better achievable rate
region. Moreover, joint decoding of data streams at respective
receivers with dual channel matching removes the adverse effect of
inter-stream interference which is caused due to independent
decoding of different data streams in \cite{bol}. Thus it is clear
that dual channel matching improves the achievable rate regions as
compared to the AF strategy of \cite{bol}.

\section{Coherent MIMO Two-Way Relay Channel Capacity Region}
\label{disc}
Combining the results from Section \ref{upbound} and Section \ref{ach},
we establish the following characterization of the scaling behavior of the
capacity region of the coherent MIMO two-way relay channel.

\begin{thm} Neglecting the ${\cal O}(1)$ term,
the capacity region of the coherent MIMO two-way relay channel is given by the convex hull of
\begin{eqnarray*}
\lim_{K \rightarrow \infty}R_{12} & \mapequal{w.p. 1} & \frac{M}{2}\log(K) \\
\lim_{K \rightarrow \infty}R_{21} & \mapequal{w.p. 1} & \frac{M}{2}\log(K)
\end{eqnarray*}
where $R_1$ and $R_2$ are the rate of information transfer
between $T_1 \rightarrow T_2$ and $T_2 \rightarrow T_1$, respectively.
\end{thm}
\begin{proof}
Follows from Theorem \ref{upperbound} and \ref{lowerbound}.
\end{proof}

\begin{figure}
\centering
\includegraphics[height= 3in]{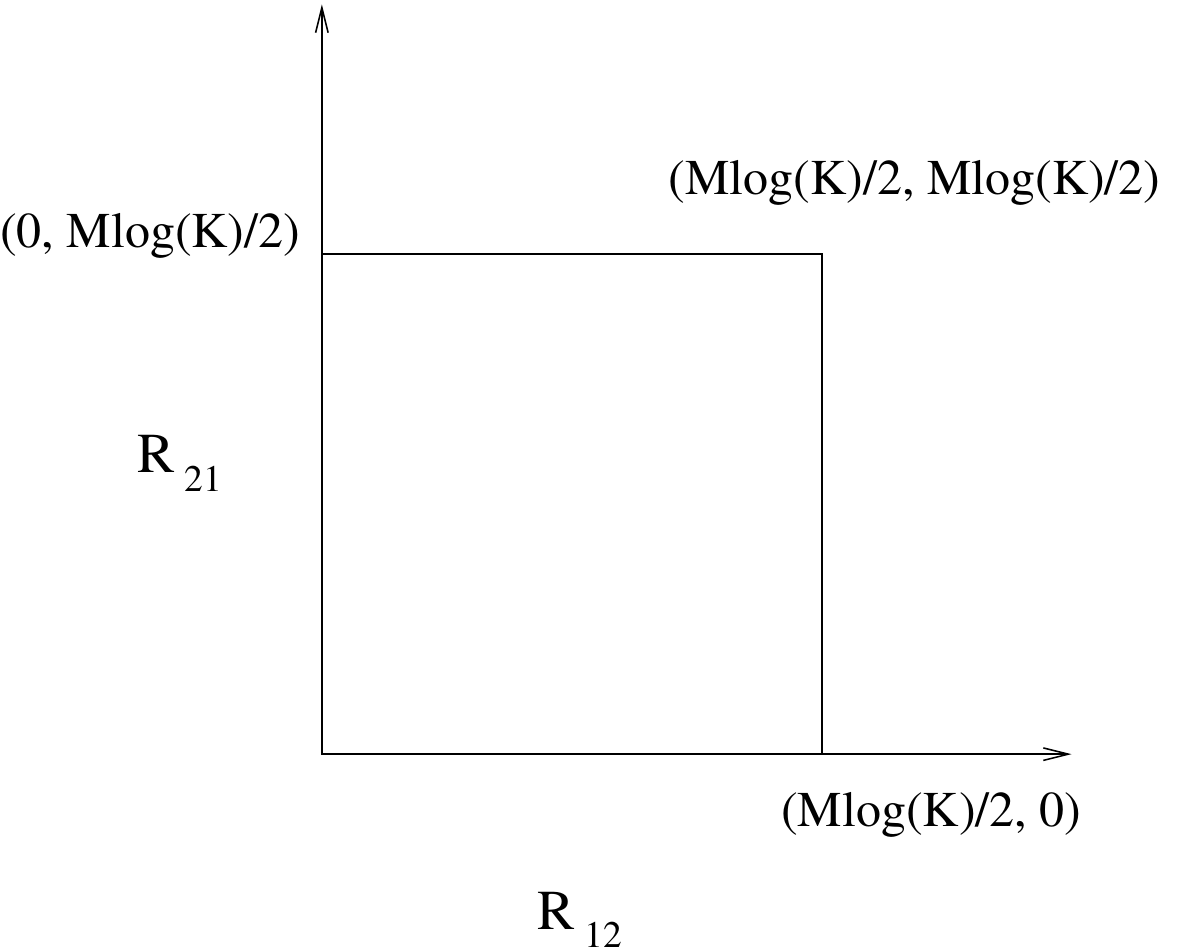}
\caption{Asymptotic Capacity of the coherent MIMO two way relaying}
\label{capplot}
\end{figure}
{\it Discussion:}
The capacity region of the coherent MIMO two-way relay channel is illustrated in
Fig. \ref{capplot}.
Combining Theorems \ref{upperbound} and \ref{lowerbound},
the sum capacity (sum of $R_{12}$ and $R_{21}$) of the coherent
MIMO two-way relay channel is given by
$M\log{(K)} + {\cal O}(1)$, as $K \rightarrow \infty$, which is exactly
double of the capacity achievable in each direction $T_1 \rightarrow T_2$
or $T_2 \rightarrow T_1$ \cite{bol}.
The ${\cal O}(1)$ term in the upper and lower bound can in general be
different and hence we characterize the exact capacity up to a
${\cal O}(1)$ term. An important implication of Theorem \ref{lowerbound}
is that with enough relays, the dual channel matching strategy
is optimal in the sense of achieving the right capacity scaling.
Therefore, neglecting the ${\cal O}(1)$ term, what this result shows is
that with the coherent MIMO two-way relay channel, one can communicate at rate $\frac{M}{2}\log{K}$ from
$T_1 \rightarrow T_2$ while simultaneously communicating at rate
$\frac{M}{2}\log{K}$ from $T_2 \rightarrow T_1$.

Recall that we obtained the lower bound in Theorem \ref{lowerbound} by
fixing $\alpha = \frac{1}{2}$ i.e. $T_1$ and $T_2$ transmit and receive for equal amount of time.
Since this lower bound is only a ${\cal O}(1)$ term away from the upper bound,
allocating equal amount of time for the transmit and the receive phase
is optimal for the coherent MIMO two-way relay channel.

From Theorem \ref{lowerbound}, it is also clear that the upper bound on
the capacity region of the MIMO two-way relay channel is achievable within
 a ${\cal O}(1)$ term without any cooperation between $T_1$ and $T_2$.
This is significant since the upper bound is for some joint encoding between $T_1$ and $T_2$. This is made possible because with channel knowledge, both $T_1$ and $T_2$ are able to cancel off the self interference.

Compared to the asymptotic capacity result for MIMO one-way relay channel
\cite{bol}, our results show that with the coherent MIMO two-way relay channel one can remove the
$\frac{1}{2}$ rate loss factor on the capacity, which comes from the half-duplex assumption on the terminals and relays. Therefore with the coherent
MIMO two-way relay channel it is possible to can achieve unidirectional full-duplex performance with half-duplex terminals.

To compute the capacity of the coherent MIMO two-way relay channel,
we assumed that CSI was available at both $T_1, T_2$ and each relay stage,
which is a very strict requirement to meet in practice.
To do dual channel matching for the coherent MIMO two-way relay channel,
the $k^{th}$ relay needs to know the realization
of ${\bf H}_k$ and ${\bf G}^{(r)}_k$ of the transmit phase and
the realization of ${\bf G}_k$ and ${\bf H}^{(r)}_k$ of the receive phase.
To cancel the self-interference
and to detect the incoming signal in the receive phase,
 both the terminals $T_1$ and $T_2$ need to know the realization of
$\{{\bf H}_k,{\bf G}_k\} \ \forall k$ for the transmit phase and
$\{{\bf H}^{(r)}_k, {\bf G}^{(r)}_k\} \ \forall k$ for the receive phase.
In practice, this is a very strict and challenging requirement, but by
sending training sequence and using standard channel estimation techniques
together with intelligent channel information feedback algorithms, all the nodes can learn
the required receive channel coefficients with good enough accuracy.

For example, by sending training sequences from $T_1$ and $T_2$,
the $k^{th}$ relay can learn ${\bf H}_k$ and ${\bf G}^{(r)}_k$
in the transmit phase (when  $T_1$ and $T_2$ transmit signals to all the relays).
Learning the ${\bf G}_k$ and ${\bf H}^{(r)}_k$ realization at the $k^{th}$ relay for the receive phase
(when all the relays transmit and both  $T_1$ and $T_2$ receive)
 is more
challenging. In time-division duplex system, however, by employing calibration
at transmitter and receiver, the forward and backward channel
can be assumed to be reciprocal in which case the
realization of ${\bf H}_k^{(r)}$ and ${\bf G}_k$ for the receive phase
is approximately equal to the realization of
${\bf H}_k^T$ and ${\bf G}^{(r)T}_k$ for the transmit phase.
Instead if a frequency-division duplex (FDD) system is used, assuming
block fading channel,
${\bf G}_k$ and ${\bf H}^{(r)}_k$ can be learnt at each relay
for receive phase, by feeding back the information about
${\bf G}_k$ and ${\bf H}^{(r)}_k$ from $T_1$ and $T_2$ in the transmit phase,
learnt in the last receive phase at  $T_1$ and $T_2$.

 To decode the incoming signal and to cancel the self interference,
$T_1$ and $T_2$ needs to know the realization of $ {\bf G}_k, {\bf
H}_k$ of the transmit phase and the realization of ${\bf G}^{(r)}_k,
{\bf H}^{(r)}_k $ of the receive phase. By sending training
sequences from all the relays to $T_1$ and $T_2$, $T_1$ and $T_2$
can learn the realization of  ${\bf H}^{(r)}_k, {\bf G}_k$ of the
receive phase, respectively. At the start of receive phase each
relay knows the realization of $ {\bf G}_k, {\bf H}_k ,  {\bf
G}^{(r)}_k, {\bf H}^{(r)}_k$, therefore if each relay transmits
quantized channel information about $ {\bf G}_k, {\bf H}_k ,  {\bf
G}^{(r)}_k, {\bf H}^{(r)}_k$ using strategies such as Grassmannian
codebook \cite{love} etc. to $T_1$ and $T_2$, both $T_1$ and $T_2$
can learn the required CSI in the receive phase.

\section{Non-Coherent MIMO Two-Way Relay Channel}
\label{noncoh} In the last section we derived the scaling behavior of the
capacity region of the MIMO two-way relay channel when both $T_1$ and $T_2$ have receive CSI while
all the relays have perfect transmit and receive CSI.
It is well known, however, that acquiring accurate CSI in a
real-time communication is a challenging problem (large overhead and
complexity) and guaranteeing near perfect CSI is almost impossible in
practice. Therefore in this section we study the scaling behavior of
the capacity region of the MIMO two-way relay channel when CSI is only available at $T_1$ and $T_2$
in the receive phase and no CSI is available at any of the
relays.
Furthermore, for this case we fix $\alpha = \frac{1}{2}$, i.e.
$T_1$ and $T_2$ transmit and receive for equal amount of time
(transmit phase is equal to receive phase).

For the non-coherent MIMO two-way relay channel, we first upper bound the
achievable rates from $T_1 \rightarrow T_2$ and $T_2 \rightarrow T_1$
using the cut-set bound for the multiple access cut. Then
using a simple AF strategy at each relay, we compute the
achievable rates from $T_1 \rightarrow T_2$ and $T_2 \rightarrow T_1$
which are shown to be within a ${\cal O}(1)$ term  from
the upper bound in the high signal to noise (SNR) regime,
thereby characterizing high SNR capacity.

\subsection {Upper Bound on The Capacity Region of The Non-Coherent Two-Way
Relay Channel}
\label{noncohup} As proved in the section \ref{upbound}, the rate of
information transfer from $T_1 \rightarrow T_2$ ($T_2 \rightarrow
T_1$) is upper bounded by the rate of information transfer between
all-relays put together and $T_2 (T_1)$ (multiple access cut). We
evaluate this upper bound in the following Theorem, when CSI is not
available at any of the relay.
\begin{thm}
\label{noncoherentdf}
In the high-SNR regime (large $P_R$), the rates $R_{12}$ and $R_{21}$ from
$T_1 \rightarrow T_2$ and $T_2 \rightarrow T_1$ for the non-coherent MIMO
two-way relay channel are upper bounded by
\begin{eqnarray*}
\lim_{K\rightarrow \infty}R_{12} & \mapright{w.p.1}& \frac{M}{2}\log{\left(P_R\right)} + {\cal O}(1)  \\
\lim_{K\rightarrow \infty}R_{21} &\mapright{w.p.1}& \frac{M}{2}\log{\left(P_R\right)} + {\cal O}(1),
\end{eqnarray*}
where $P_R$ is the total power constraint across all relays.
\end{thm}
\begin{proof}
Using the multiple access cut-set bound from (\ref{t1t2mac}) and (\ref{t2t1mac}), we have
\[R_{12} \le I({\bf t}_1, {\bf t}_2, \ldots, {\bf t}_K ; {\bf y})\] and
\[R_{21} \le I({\bf t}_1, {\bf t}_2, \ldots, {\bf t}_K ; {\bf v}),\]
for some joint distribution $p\left({\bf t}_1, {\bf t}_2, \ldots, {\bf t}_K\right)$ and with no CSI at any relay.
Recall from (\ref {t2rx}) that the received signal ${\bf y}$ is
given by
\[{\bf y} = \sum_{k=1}^{K}\sqrt{\gamma_kP_k}{\bf G}_k{\bf t}_k + {\bf z}
\]
with power constraint $\sum_{k=1}^K\gamma_k\le P_R$.
Using the capacity result from Section $4.1$ \cite{tel}
for no transmit CSI
\[I({\bf t}_1, {\bf t}_2, \ldots, {\bf t}_K ; {\bf y})
 \le \log\det\left({\bf I}_M + \frac{{\bf \Sigma Q\Sigma}^{*}}{\sigma^2}\right)\]
where
\[{\bf \Sigma} = [\sqrt{P_1}{\bf G}_1 \sqrt{P_2}{\bf G}_2 \ldots \sqrt{P_K}{\bf G}_K] \in \bbC^{M\times
NK}\] and
${\bf Q}$ is the covariance matrix of \[[\sqrt{\gamma_1}{\bf t}_1 \sqrt{\gamma_2}{\bf t}_2 \ldots \sqrt{\gamma_K}{\bf t}_K]^T \in \bbC^{NK\times 1}\] with
$[{\bf t}_1 \ {\bf t}_2 \  \ldots \ {\bf t}_K]^T$ circularly symmetric
complex Gaussian and equivalent power constraint of $tr({\bf Q})\le P_R$
and the maximum is achieved when
${\bf Q} = \frac{P_R}{NK}{\bf I}_{NK\times NK}$.
Therefore, using ${\bf Q} = \frac{P_R}{NK}{\bf I}_{NK\times NK}$
\[R_{12} \le \log\det\left({\bf I}_M + \frac{P_R}{NK\sigma^2}\sum_{k=1}^KP_k{\bf G}_k{\bf G}_k^{*}\right).\]
From the strong law of large numbers,
\[\lim_{K \rightarrow \infty}\frac{1}{K}\sum_{k=1}^KP_k{\bf G}_k{\bf G}_k^{*}
\xrightarrow{w.p.1} \bbE\left\{P_k{\bf G}_k{\bf G}_k^{*}\right\} =\bbE\left\{P_k\right\}\bbE\left\{{\bf G}_k{\bf G}_k^{*}\right\}.\]
Since $\bbE\left\{{\bf G}_k{\bf G}_k^{*}\right\} = N{\bf I}_M$ and
let $\mu\bydef \bbE\left\{P_k\right\}$,
\[R_{12} \mapright{w.p.1} \log\det\left({\bf I}_{M} + \frac{P_R\mu}{\sigma^2}{\bf I}_M\right).\]
Since $\mu$ and $\sigma^2$ are finite, for large $P_R$
\[R_{12} \mapright{w.p.1} M\log{P_R} + {\cal O}(1).\]
Since $T_1$ and $T_2$ transmit only
for half the time ($\alpha = \frac{1}{2}$) in any given time slot
\[R_{12} \mapright{w.p.1} \frac{M}{2}\log{P_R} + {\cal O}(1).\]
Similarly, it can be shown that
\[R_{21} \mapright{w.p.1} \frac{M}{2}\log{P_R} + {\cal O}(1).\]
\end{proof}
\subsection {Lower Bound on The Capacity Region of The Non-coherent MIMO Two-Way Relay Channel}
In this subsection we compute achievable rates $R_{12}$ and $R_{21}$
for the non-coherent MIMO two-way relay channel using a simple AF
strategy at each relay. The strategy is the following:
with no CSI at any relay,
each relay just normalizes the received signal to meet its power constraint
and retransmits it in the receive phase. With CSI available at each
destination $T_1$ ($T_2$), self interference generated by $T_1$ ($T_2$)
is removed from the received signal and the equivalent channel between
$T_1 \rightarrow T_2$ ($T_2 \rightarrow T_1$)
for the non-coherent MIMO two-way relay channel is given by
$\sum_{k=1}^K{\bf H}_k{\bf G}_k$ $\left(\sum_{k=1}^K{\bf H}^{(r)}_k{\bf
G}^{(r)}_k\right)$. As $K\rightarrow \infty$, this channel is shown to
behave as i.i.d. MIMO Gaussian channel. Then by using the
capacity results from \cite{tel}, we lower bound the capacity region of the
non-coherent MIMO two-way relay channel. We show that with approximately
same power used at $T_1 (T_2)$ and all relays (i.e. $P\approx P_R$), the
lower bound meets the upper bound in the high SNR regime (high $P$).

\label{noncohaf}
The following Theorem gives the expressions for achievable $R_{12}$ and $R_{21}$
pair, when each relay uses AF.
\begin{thm}
\label{noncoherent}
In the high SNR regime, the achievable rate region
for the non-coherent MIMO two-way relay channel using AF
strategy at each relay, is given by
\begin{eqnarray*}
\lim_{K\rightarrow\infty}R_{12} &\mapequal{w.p.1}& \frac{M}{2}\log{\left(P_R\right)} + {\cal O}(1) \\
\lim_{K\rightarrow\infty}R_{21} &\mapequal{w.p.1}& \frac{M}{2}\log{\left(P_R\right)} + {\cal O}(1).
\end{eqnarray*}
\end{thm}
\begin{proof}
Recall from (\ref{relayrx}) that the received signal at each relay is
given by
\begin{equation}
{\bf r}_k =  \sqrt{\frac{PE_k}{M}} {\bf H}_{k}{\bf x} + \sqrt{\frac{PF_k}{M}}{\bf G}_{k}^{(r)}{\bf u} + {\bf n}_k
\end{equation}
Therefore the average received signal plus noise power at each relay
 is given by $N(P(E_k + F_k) + \sigma^2)$.
We assume that the $k^{th}$ relay knows the average received signal plus noise
power $N(P(E_k + F_k) + \sigma^2)$ and transmits
${\bf t}_k = \left(\frac{1}{N(P( E_k + F_k) +
\sigma^2))}\right)^{\frac{1}{2}}{\bf r}_k$ to ensure that
${\bbE}\{{\bf t}_k^*{\bf t}_k\} = 1$.
With this normalization, from (\ref{t1rx}) and (\ref{t2rx}),
the received signal at terminal $T_1$ and $T_2$ is given by ${\bf v}$ and
${\bf y}$, respectively, where
\[{\bf v} =  \sum_{k=1}^K\sqrt{\frac{\gamma_kQ_k}{N(P(E_k+ F_k)+\sigma^2)}}{\bf H}_k^{(r)}
{\bf r}_k + {\bf w}\]
\[{\bf y} = \sum_{k=1}^K\sqrt{\frac{\gamma_kP_k}{N(P(E_k+ F_k)+\sigma^2)}}{\bf G}_k{\bf r}_k + {\bf z}.\]
Substituting for ${\bf r}_k$ from (\ref{relayrx}) in the above equation
\begin{eqnarray*}
{\bf y} & = & \sum_{k=1}^K\sqrt{\frac{\gamma_kPP_kE_k}{NM(P(E_k+ F_k)+\sigma^2)}}{\bf G}_k{\bf H}_k{\bf x} \\
& & +\sum_{k=1}^K\sqrt{\frac{\gamma_kPP_kF_k}{NM(P(E_k+ F_k)+\sigma^2)}}{\bf G}_k{\bf G}_k^{(r)}{\bf u}\\
& & +\sum_{k=1}^K\sqrt{\frac{\gamma_kP_k}{N(P(E_k+ F_k)+\sigma^2)}}{\bf G}_k{\bf n}_k + {\bf z}.
\end{eqnarray*}
Since $T_2$ knows ${\bf u}$ and has perfect CSI, it can cancel the self
interference. Removing the self interference from ${\bf y}$ and dividing both sides by $\sqrt{K}$,
\begin{eqnarray*}
\label{tx2noncon}
{\bf y}' & = & \underbrace{ \frac{1}{\sqrt{K}} \sum_{k=1}^K\sqrt{\frac{\gamma_kPP_kE_k}{NM(P(E_k+ F_k)+\sigma^2)}}{\bf G}_k{\bf H}_k}_{{\bf A}}{\bf x} \\
& & +\underbrace{\frac{1}{\sqrt{K}}\sum_{k=1}^K\sqrt{\frac{\gamma_kP_k}{N(P(E_k+ F_k)+\sigma^2)}}{\bf G}_k{\bf n}_k + \frac{1}{\sqrt{K}}{\bf z}}_{{\bf b}}.\\
\end{eqnarray*}
Similarly $T_1$ knows ${\bf x}$ and also has perfect CSI, therefore it can
also remove the self interference.
Removing the self interference from ${\bf v}$ and dividing both sides by $\sqrt{K}$
\begin{eqnarray*}
\label{tx2noncon}
{\bf v}' & = & \underbrace{\frac{1}{\sqrt{K}}\sum_{k=1}^K\sqrt{\frac{\gamma_kPQ_kF_k}{NM(P(E_k+ F_k)+\sigma^2)}}{\bf H}_k^{(r)}{\bf G}_k^{(r)}}_{{\bf C}}{\bf u} \\
& & +\underbrace{\frac{1}{\sqrt{K}}\sum_{k=1}^K\sqrt{\frac{\gamma_kQ_k}{N(P(E_k+ F_k)+\sigma^2)}}{\bf H}_k^{(r)}{\bf n}_k + \frac{1}{\sqrt{K}}{\bf w}}_{{\bf d}}.\\
\end{eqnarray*}
As $K \rightarrow \infty$, it can be shown that (Theorem 3 \cite{bol})
\[{\bf A}_{i,j} \sim {\cal CN}\left (0, \frac{1}{K}\sum_{k=1}^K{\bbE}\left\{\frac{\gamma_kPP_kE_k}{M(P(E_k+ F_k)+\sigma^2)}\right\}\right)\]
\[{\bf C}_{i,j} \sim {\cal CN}\left (0, \frac{1}{K}\sum_{k=1}^K{\bbE}\left\{\frac{\gamma_kPQ_kF_k}{M(P(E_k+ F_k)+\sigma^2)}\right\}\right)\]
\[{\bf b}_{i} \sim {\cal CN}\left (0, \frac{\sigma^2}{K}\left(\sum_{k=1}^K{\bbE}\left\{\frac{\gamma_kP_k}{(P(E_k+ F_k)+\sigma^2)}\right\}+1\right)\right)\]
\[{\bf d}_{i} \sim {\cal CN}\left (0, \frac{\sigma^2}{K}\left(\sum_{k=1}^K{\bbE}\left\{\frac{\gamma_kQ_k}{(P(E_k+ F_k)+\sigma^2)}\right\}+1\right)\right)\]
and
\[R_{\bf A} \xrightarrow{w.p. 1} \frac{1}{K}\sum_{k=1}^K{\bbE}\left\{\frac{\gamma_kPP_kE_k}{M(P(E_k+ F_k)+\sigma^2)}\right\}{\bf I}_{M^2}\]
\[R_{\bf C} \xrightarrow{w.p. 1} \frac{1}{K}\sum_{k=1}^K{\bbE}\left\{\frac{\gamma_kPQ_kF_k}{M(P(E_k+ F_k)+\sigma^2)}\right\}{\bf I}_{M^2}\]
where ${\bf A}_{i,j}, {\bf C}_{i,j}$ denotes $i^{th}$ row and $j^{th}$ column
entry of
${\bf A}$ and ${\bf C}$ respectively and ${\bf b}_{i}, {\bf d}_{i}$ denotes the $i^{th}$ element of
${\bf b}$ and ${\bf d}$ respectively, $R_{\bf A} = \bbE\{{\bf a}{\bf a}^*\}$ where ${\bf a} = {\text vec}({\bf A})$ and $R_{\bf C} = \bbE\{{\bf c}{\bf c}^*\}$ where ${\bf c} = {\text vec}({\bf C})$.

This shows that the channel matrices ${\bf A}, {\bf C}$ and the noise vectors
${\bf b}, {\bf d}$ are i.i.d. Gaussian, therefore using results from
Section $4.1$ \cite{tel} with only receive CSI and no transmit CSI,
the achievable rate $R_{12}$ ($R_{21}$) of the $T_1 \rightarrow T_2$
($T_2 \rightarrow T_1$) link for $\alpha = \frac{1}{2}$, is given by
\[\lim_{K \rightarrow \infty}R_{12} \mapequal{w.p. 1} \ \frac{1}{2}{\bbE}_{{\bf H}_w}\left\{\log \det\left({\bf I}_M + \frac{\rho_1}{M}{{\bf H}_w}{\bf H}_w^*\right)\right\} \]

\[\lim_{K \rightarrow \infty}R_{21} \mapequal{w.p. 1} \ \frac{1}{2}{\bbE}_{{\bf H}_w}\left\{\log \det\left({\bf I}_M + \frac{\rho_2}{M}{{\bf H}_w}{\bf H}_w^*\right)\right\}\]
where ${\bf H}_w$ is an $M \times M$ matrix with i.i.d.
${\cal CN}(0,1)$ entries and
\[\rho_1 = \frac{\frac{1}{K}\sum_{k=1}^K{\bbE}\left\{\frac{\gamma_kPP_kE_k}{(P(E_k+F_k)+\sigma^2)}\right\}}{\frac{\sigma^2}{K}(\sum_{k=1}^K{\bbE}\left\{\frac{\gamma_kP_k}{(P(E_k+F_k) + \sigma^2)}\right\}+1)},\]
\[\rho_2 = \frac{\frac{1}{K}\sum_{k=1}^K{\bbE}\left\{\frac{\gamma_kPQ_kF_k}{(P(E_k+F_k)+\sigma^2)}\right\}}{\frac{\sigma^2}{K}(\sum_{k=1}^K{\bbE}\left\{\frac{\gamma_kQ_k}{(P(E_k+F_k) + \sigma^2)}\right\}+1)}.\]
Note that $\rho_1$ and $\rho_2$ are effective SNRs.
Denoting $\mu =  {\bbE}\{E_k\} =  {\bbE}\{F_k\} = {\bbE}\{P_k\} =  {\bbE}\{Q_k\}\ \forall k$, and $\frac{1}{\eta} = {\bbE}\{\frac{E_k}{P(E_k+F_k) + \sigma^2}\}
 = {\bbE}\{\frac{F_k}{P(E_k+F_k) + \sigma^2}\} \ \forall k$,
\[\rho_1 = \rho_2 = \frac{\frac{P\mu}{\eta}\sum_{k=1}^K\gamma_k}
{\sigma^2(\frac{1}{\eta}\sum_{k=1}^K\gamma_k + 1)}.\]
Since the relay power is constrained by $\sum_{k=1}^K\gamma_k = P_R$
\[\rho_1 = \rho_2 = \frac{PP_R\mu}{\sigma^2(P_R\eta_2 + \eta )}.\]
Choosing $P \approx P_R$,
\[\rho_1 = \rho_2 \approx \frac{P_R}{\sigma^2}\] since $E_k, F_k, P_k, Q_k.$
$\forall k$ are bounded.
Therefore
\[\lim_{K \rightarrow \infty}R_{12} = \lim_{K \rightarrow \infty} R_{21} \mapequal{w.p. 1}
{\bbE}_{{\bf H}_w}
\left\{\log \det\left({\bf I}_M + \frac{P_R}{M\sigma^2}{{\bf H}_w}{\bf H}_w^*
\right)\right\}.\]
In high SNR regime $P\approx P_R \gg 1$, from \cite{tel}, it follows that
\begin{eqnarray*}
\lim_{K \rightarrow \infty}R_{12} & \mapequal{w.p. 1} & \frac{M}{2}\log{\left(P_R\right)} + {\cal O}(1) \\
\lim_{K \rightarrow \infty}R_{21} &\mapequal{w.p. 1} &  \frac{M}{2}\log{\left(P_R\right)} + {\cal O}(1).
\end{eqnarray*}
\end{proof}

{\it Discussion:} In this section, we first obtained an upper bound
on the capacity region of the non-coherent MIMO two-way relay channel
using multiple access cut-set bound when CSI is only known at $T_1$ and $T_2$.
Then with the help of a simple AF strategy we provided a
lower bound which is a ${\cal O}(1)$ term away from the upper bound in the high SNR regime.  We find that, contrary to the coherent case, with
the non-coherent MIMO two-way relay channel, as the number of relay nodes
grow large, the capacity region expression is independent of the number of
relays and no coherent combining gain (array gain) is available when
there is no CSI at any relay. Similar to the coherent case, however,
it turns out that even in the non-coherent case both $T_1$ and $T_2$ can
simultaneously transmit at a rate which is equal to the maximum
possible rate at which they could have transmitted when there is no
data flowing in the opposite direction. Therefore, the non-coherent
MIMO two-way relay channel creates two
orthogonal channels, one from $T_1 \rightarrow T_2$ and another from
$T_1 \rightarrow T_2$ with rate $\frac{M}{2}\log{P_R}$ achievable on each link
simultaneously, thereby removing the $\frac{1}{2}$ rate loss factor
because of half-duplex nodes.

The lower bound provided by Theorem \ref{noncoherent} shows that
the achievable rate for the non-coherent MIMO two-way relay channel
is same as the capacity of a point to point $M \times M$ i.i.d. Gaussian channel with receive SNR $P_R$,
with perfect CSI at receiver and no CSI at transmitter and where $\frac{1}{2}$ factor is due to the half-duplex
requirement.
This result is quite intuitive, since with absence of CSI at the relays,
as $K \rightarrow \infty$ the equivalent channel between
 $T_1 \rightarrow T_2$ ($T_2 \rightarrow T_1$) converges to
an $M \times M$ i.i.d. Gaussian channel and therefore
the result follows from \cite{tel}.

Compared to (Theorem 3 \cite{bol}), this result shows that
with the non-coherent MIMO two-way relay channel
it is possible to remove the $\frac{1}{2}$ rate loss factor due to the half-duplex
constraint and can achieve the same rate as promised by Theorem 3 \cite {bol}
(for unidirectional communication), in each direction $T_1 \rightarrow T_2$ and $T_2
\rightarrow T_1$. This is again due to the fact that, with perfect
CSI both $T_1$ and $T_2$ can cancel the
self interference terms their own transmitted signals generate and hence the
received signal at $T_2$ ($T_1$) when $T_2$ is also sending information
is equivalent to the received signal at $T_2$ in \cite{bol}, where there is no communication happening on
$T_2 \rightarrow T_1$ link.
Therefore there is a two-fold increase in achievable rate with
the non-coherent MIMO
two-way relay channel in comparison to \cite{bol}.

\section{Conclusion}
\label{conc} In this paper we developed capacity scaling laws for
 the MIMO two-way relay channel under coherent and non-coherent
assumptions. First we upper bounded the capacity region of the
coherent MIMO two-way relay channel
using the broadcast and multiple access cut-set
bound. Then we proposed a dual channel matching strategy to obtain an
achievable rate region for the coherent MIMO two-way relay channel.
The achievable rate region was shown to be
a ${\cal O}(1)$ term away from the upper bound, as $K \rightarrow
\infty$. Hence we characterized the coherent MIMO two-way relay channel
capacity region within a ${\cal O}(1)$ term as  $K \rightarrow
\infty$.

The dual channel matching strategy we proposed for the coherent MIMO
two-way relay channel is a decentralized strategy, where each relay
node does not cooperate with any other relay node and
only uses its CSI to coherently match the channels which the
streams from $T_1$ and $T_2$ experience. An interesting outcome of
our analysis is that the dual channel matching strategy, which
requires no cooperation between relays, achieves the capacity region
upper bound which allows for full cooperation between relays, within
a ${\cal O}(1)$ term. Thus, dual channel matching not only
simplifies the practical protocol design, but also achieves capacity
region upper bound within a ${\cal O}(1)$ term.

For the coherent MIMO two-way relay channel, there is a strict
requirement that all the nodes need to know perfect CSI, which in practice can
be quite challenging and resource consuming. Therefore we also
considered the case when only $T_1$ and $T_2$ have perfect receive
CSI and none of the relays have any CSI, which is referred to as the
non-coherent MIMO two-way relay channel. For this case we upper bounded
the capacity region using only the multiple access cut-set bound and
fixing $\alpha = \frac{1}{2}$ (i.e. $T_1, \ T_2$ and all the relays
transmit for equal amount of time in each time slot). Then with the
help of a simple AF strategy, we showed that in the high SNR regime
the upper bound is achievable within a ${\cal O}(1)$ term, and hence
characterize high SNR capacity region of the non-coherent MIMO two-way relay
channel. Thus, we showed that a very simple AF strategy which transmits
the power normalized version of the received signal is an optimal strategy.
The intuition behind this result is that by using AF with no CSI at any relay,
the effective channel from the source to destination $T_1\rightarrow T_2$ or $T_2\rightarrow T_1$ converges to an $M \times M$ i.i.d. MIMO Gaussian channel
as $K\rightarrow \infty$, which is similar to the effective channel considered
for the capacity region upper bound (using the multiple access cut). The
upper and lower bound differ by ${\cal O}(1)$ term because of the forwarded
noise from each relay node.

Compared to \cite{gast, bol}, our capacity scaling results
for the coherent and non-coherent MIMO two-way relay channel
shows that with the MIMO two-way relay channel there is a two-fold increase
in the capacity than unidirectional communication with large number of relays.
Hence, the MIMO two-way relay channel helps in improving the
spectral efficiency and unidirectional
full-duplex performance while using half-duplex terminals.


\begin{thebibliography}{10}

\bibitem{van} E.C. Van der Meulen, ``Three terminal communication
channels," \emph{Adv. Appl. Probab.}, vol. 3, pp. 120-154, 1971.

\bibitem{cover1} T.M. Cover and A.A. El Gamal, ``Capacity theorems for
relay channels," \emph{IEEE Trans. on Information Theory}, vol. 25, no.
5, pp. 572-584, Sept 1979.

\bibitem{wang} B. Wang, J. Zhang and A. Host-Madsen, ``On the capacity
of MIMO relay channels, " \emph{IEEE Trans. on Information Theory},
vol. 51, no. 1, pp. 29- 43, Jan. 2005.

\bibitem{sen1} A. Sendonaris, E. Erkip, B Aazhang, ``User cooperation
diversity. Part I. System description," \emph{IEEE Trans. on
Communications}, vol. 51, no. 11, pp. 1927- 1938, Nov. 2003.

\bibitem{sen2} A. Sendonaris, E. Erkip, B Aazhang, ``User cooperation
diversity. Part II. Implementation aspects and performance analysis,"
\emph{IEEE Trans.  on Communications}, vol. 51, no. 11, pp. 1939- 1948,
Nov. 2003.

\bibitem{nabar} R.U. Nabar, H. B\"olcskei, F.W. Kneubuhler, ``Fading
relay channels: performance limits and space-time signal design,"
\emph{IEEE Journal on Selected Areas in Communications}, vol. 22, no. 6, pp.
1099- 1109, Aug. 2004.

\bibitem{lane} J. Laneman, D.N.C. Tse and G.W. Wornell, ``Cooperative
diversity in wireless networks: Efficient protocols and outage
behavior," \emph{IEEE Trans. on Information Theory}, vol. 50, no. 12,  pp. 3062-
3080, Dec. 2004.

\bibitem{ala} S. M. Alamouti, ``A simple transmit diversity technique
for wireless communications ," \emph{IEEE Journal on Selected Areas in
Communications}, vol. 16, no. 8, pp. 1451-1458, Oct 1998.

\bibitem{tsc} V. Tarokh, N.Seshadri and A.~R.~Calderbank,
``Space-Time block codes for high data rate wireless
communication:Performance criterion and code construction,'' {\it IEEE Trans. on
Information Theory}, vol. 44, no. 2, pp. 744-765, Mar 1998.

\bibitem{caleb}C.K. Lo, S. Vishwanath and R.W. Heath, Jr., ``
Rate bounds for MIMO relay channels using precoding," in \emph{Proc.
IEEE GLOBECOM}, vol. 3, pp. 1172-1176, St. Louis, MO, Nov. 2005.
\bibitem{host} A. Host-Madsen, J. Zhang, ``Capacity bounds and power
allocation for wireless relay channels," \emph{IEEE Trans. on
Information Theory}, vol. 51, no. 6,  pp. 2020- 2040, June 2005.

\bibitem{gupta} P. Gupta, P.R. Kumar, ``The capacity of wireless
networks,"
\emph{IEEE Trans. on Information Theory}, vol. 46, no. 2, pp. 388-404,
Mar. 2000.

\bibitem{gupta1} G. Kramer, M Gastpar and P. Gupta, ``Cooperative Strategies and Capacity Theorems for Relay Networks,"\emph{IEEE Trans. on Information Theory}, vol. 51, no. 9, pp. 3037-3063,Sept. 2005.

\bibitem{sriram} S Vishwanath, S Jafar, S Sandhu, ``Half-duplex relays: cooperative communication strategies and outer bounds," in {\emph Proc. International Conference on Wireless Networks, Communications and Mobile Computing},
vol. 2, pp. 1455 - 1459, 13-16 June 2005.

\bibitem{lane1}
J.N. Laneman, G.W. Wornell, ``Distributed Space-Time-Coded Protocols
for Exploiting Cooperative Diversityin Wireless Networks,"\emph{IEEE Trans. on Information Theory}, vol. 49, no. 10, pp. 2415-2425, Oct 2003.

\bibitem{host1}
A.H. Maden, ``On the Capacity of Wireless Relaying,"
in \emph{Proc. Vehicular Technology Conference, 2002. VTC 2002-Fall. 2002}
IEEE 56th, vol.3, pp. 1333-1337, 2002.

\bibitem{erkip}
M. Yuksel and E. Erkip, ``Cooperative Wireless Systems: A Diversity-Multiplexing Tradeoff Perspective," Submitted to IEEE Transactions on Information Theory Sep 2006, available online at {\it http://arxiv.org/abs/cs/0609122}.

\bibitem{valenti} S. Wei, D.L. Goeckel, and M.C. Valenti, ``Asynchronous cooperative diversity," "\emph{IEEE Trans. on Wireless Communications}, vol. 5, no. 6, pp. 1547-1557, June 2006.






\bibitem{gast} M. Gastpar and M. Vetterli, ``On the Capacity of
Wireless Networks: The Relay Case,'' in \emph{Proc. IEEE INFOCOM},
vol. 3, pp. 1577- 1586, New York June 2002.

\bibitem{bol} H. B\"olcskei, R. U. Nabar, \"O. Oyman, and A. J.
Paulraj, ``Capacity scaling laws in MIMO relay networks,'' \emph{IEEE Trans.
on Wireless Communications}, vol. 5, no. 6, pp. 1433-1444, June 2006.

\bibitem{wit} B. Rankov and A. Wittneben, ``Spectral Efficient
Signaling for Half-duplex Relay Channels,'' in \emph{Proc. Asilomar Conference
on Signals, Systems, and Computers 2005}, pp. 1066-1071, Pacific
Grove, CA, Oct-Nov. 2005.

\bibitem{wit2} B. Rankov and A. Wittneben, ``Achievable Rate Regions
for the Two-way Relay Channel,'' in \emph{Proc. IEEE Int. Symposium on
Information Theory (ISIT)}, pp. 1668-1672, Seattle, USA, July 2006.

\bibitem{hol} T.J Oechtering and H. Boche, ``Optimal Tranmsit Strategies
In Multi-Antenna Bidirectional Relaying," in \emph{Proc. IEEE Intern. Conf. on Acoustics, Speech, and Signal Processing (ICASSP '07)}, vol. 3, pp. 145-148, Honolulu, Hawaii, USA, Apr. 2007.

\bibitem{tar} S. Joon Kim, P. Mitran and V. Tarokh, ``Performance Bounds
for Bi-Directional Coded Cooperation Protocols." available at
{\emph http://arxiv.org/PS\_cache/cs/pdf/0703/0703017v1.pdf}.

/and

\bibitem{ger} G. Kramer and S.A. Savari, ``On networks of two-way
channels," \emph{ Algebraic Coding Theory and Information Theory, DIMACS
Workshop, Dec. 15-18, 2003}, Rutgers University, DIMACS Series in Discrete
Mathematics and Theoretical Computer Science, vol. 68, A. Ashikhmin and
A. Barg, eds., pp. 133-143, available at {\it
http://cm.bell-labs.com/who/gkr/}.

\bibitem{kra} G. Kramer, ``Models and theory for relay channels with receive constraints," in \emph{Proc. 42nd Annual Allerton Conf. on Commun. Control and Comp}, (Monticello, IL, USA), pp. 1312-1321, Sept. 29-Oct. 1, 2004.

\bibitem{ray} R. Ahlswede, N. Cai, S.-Y. R. Li and R. W. Yeung,
``Network information flow," \emph{IEEE Trans. on Information Theory},
vol. 46, no. 4, pp. 1204-1216, July 2000.

\bibitem{dina} S. Katti, R. Hariharan, W. Hu, D. Katabi, M. Medard, and J. Crowcroft, ``XORs In The Air: Practical Wireless Network Coding," in \emph{Proc. ACM SIGCOMM}, pp. 243-254, Pisa, Italy, Sept. 2006.



\bibitem{tel} I.E. Telatar, ``Capacity of Multi-Antenna Gaussian
Channels,'' \emph{European Trans. on Telecommunications}, vol. 10, no. 6, pp.
585-595, Nov./Dec. 1999.

\bibitem{cover} T. M. Cover, J. A. Thomas, \emph{Elements of
Information Theory}, John Wiley and Sons 2004.

\bibitem{medina} O. Munoz-Medina, J. Vidal, A. Agustín,
``Linear Transceiver Design in
Nonregenerative Relays With Channel State Information", IEEE Trans. on
Signal Processing,
Vol. 55, pp. 2593-2604, June 2007.
\bibitem{love}
D.J. Love, R.W. Heath Jr. and T. Strohmer,
``Grassmannian beamforming for multiple-input multiple-output wireless
systems", \emph{IEEE Trans. on Information Theory}, vol. 49, no. 10,
pp. 2735-2747, Oct. 2003.



\end{thebibliography}
\end{document}